\documentclass[12pt]{article}
\usepackage[utf8]{inputenc}
\usepackage[letterpaper]{geometry}
\usepackage{color}
\usepackage{float}
\usepackage{mathrsfs}
\usepackage{url}
\usepackage{bm}
\usepackage{hepnames}
\usepackage{amsmath}
\usepackage{amssymb}
\usepackage{graphicx}
\usepackage{setspace}
\usepackage{esint}
\usepackage[authoryear,longnamesfirst]{natbib}
\usepackage[bookmarks=true,bookmarksnumbered=false,bookmarksopen=true,bookmarksopenlevel=1,
 breaklinks=true,pdfborder={0 0 1},backref=false,colorlinks=true]
 {hyperref}
\hypersetup{
 citecolor=blue,urlcolor=blue,linkcolor=blue}

\makeatletter

\providecommand{\tabularnewline}{\\}

\usepackage{color}
\usepackage{url}
\usepackage{bm}
\usepackage{hepnames}

\providecommand{\tabularnewline}{\\}

\usepackage{amsfonts}\usepackage{mathrsfs}\usepackage{bm}\usepackage{dsfont}\usepackage{tabularx}
\usepackage{lipsum}\usepackage{ltablex}\usepackage{array}\usepackage{makecell}\usepackage{caption}\usepackage{multirow}\usepackage{wasysym}\usepackage{soul}\usepackage{rotating}
\usepackage[hang,flushmargin]{footmisc}
\usepackage{placeins}
\usepackage{algorithm}
\usepackage{algpseudocodex}
\normalbaselines
\providecommand{\U}[1]{\protect\rule{.1in}{.1in}}
\newtheorem{theorem}{Theorem}[section]
\newtheorem{assumption}{Assumption}[section]
\newtheorem{corollary}[theorem]{Corollary}\newtheorem{lemma}[theorem]{Lemma}\newtheorem{proposition}[theorem]{Proposition}\newtheorem{remark}{Remark}[section]
\newenvironment{proof}[1][Proof]{\noindent\if\relax\detokenize{#1}\relax\else\textbf{#1.} \fi}{\ \rule{0.5em}{0.5em}}

\numberwithin{equation}{section}

\newcommand{\Keywords}[1]{\par\noindent {\small{\em Keywords\/}: #1}}
\newcommand{\Ee}{\mathbb{E}}
\newcommand{\ind}{\mathds{1}}
\newcommand{\R}{\mathbb{R}}
\newcommand{\dCov}{\operatorname{dCov}}
\renewcommand{\Pp}{\Pr}
\newcommand{\Exp}{\operatorname{Exp}}

\makeatother

\begin{document}
\title{\textbf{Testing Tail-Index Homogeneity via Exceedance--Covariate
Independence}}
\author{Thomas T. Yang\thanks{Email: \protect\protect\protect\protect\protect\protect\protect\protect\protect\protect\protect\protect\protect\protect\protect\protect\protect\protect\protect\protect\protect\protect\protect\protect\protect\protect\protect\protect\protect\protect\protect\protect\protect\protect\protect\protect\protect\protect\protect\protect\protect\href{mailto:tao.yang@anu.edu.au}{tao.yang@anu.edu.au}.
Replication Package: \protect\protect\protect\protect\protect\protect\protect\protect\protect\protect\protect\protect\protect\protect\href{http://github.com/youngtigers/tail_index_test}{github.com/youngtigers/tail\_index\_test}}\\
 Australian National University}
\date{\today}

\maketitle
 
\begin{abstract}
We propose a test of whether the tail index of a heavy-tailed outcome
is homogeneous across observed covariates. The key observation is
that, under tail-index homogeneity, normalized log exceedances are
asymptotically independent of the covariates even when tail frequencies
vary with them. This leads to a simple-to-implement permutation test
based on distance covariance, without estimating conditional tail
indices or specifying the form of tail-index heterogeneity. For the
global and locally balanced versions, we establish validity for order-statistic
thresholds and give conditions for consistency and local power. Local
balancing selects extremes within two covariate groups and can improve
power when tail selection concentrates the covariates. Simulations
illustrate these gains. In French motor insurance data, upper conditional
quantiles vary with risk characteristics, but the tail-index tests
do not reject a common tail index.

\vspace{4mm}
 \medskip{}
\end{abstract}
\newpage{} \setstretch{1.7}

\section{Introduction}

\label{sec:introduction}

The tail index governs tail-probability decay, the existence of higher-order
moments, and extreme-quantile extrapolation \citep{deHaanFerreira2006}.
Classical estimators such as \citet{Hill1975} treat it as a scalar
feature of a homogeneous population. In finance, environmental studies,
and insurance, however, tail thickness may depend on observed characteristics.
We ask whether covariates change the frequency or scale of extremes,
or also their tail-decay rate. Our application to French motor third-party
liability claims illustrates this distinction: BonusMalus is positively
associated with upper conditional claim quantiles, and vehicle power
is significant at the 0.95 quantile, yet the tests do not reject a
common tail index. Covariate effects on large claims therefore need
not imply differences in tail thickness.

Tail-index homogeneity is also an important assumption behind standard
methods for extreme conditional quantiles. The heavy-tailed specialization
of the extremal quantile-regression theory of \citet{Chernozhukov2005}
and the extrapolation framework of \citet{WangLiHe2012} impose a
common tail index. If the tail index varies, these extrapolations
are misspecified, and the error grows with the extrapolation range.
For example, imposing a tail exponent that is too large for a high-risk
group can severely understate its extreme losses and the capital needed
to cover them. A test of homogeneity is therefore a natural diagnostic
before applying them.

To formalize the tail behavior, consider the conditional regular-variation
model 
\[
\bar{F}(y\mid x):=\Pr(Y>y\mid X=x)=\kappa(x)y^{-\alpha(x)}L(y;x),\qquad y\ \text{large enough},
\]
where $\alpha(x)>0$, $\kappa(x)>0$, $L(\cdot;x)$ is slowly varying
in its first argument, and normalized so that $L(y;x)\to1$ as $y\to\infty$.
Without normalization, $\tilde{L}(y;x)=\kappa(x)L(y;x)$ is a general
slow varying function. We test 
\[
H_{0}:\ \alpha(x)=\alpha_{0}\quad\text{for almost every }x
\]
against general covariate dependence of $\alpha(x)$. The null allows
the leading tail-frequency term $\kappa(x)$ and second-order tail
behavior to vary with $x$.

The core idea is that, under $H_{0}$, the log normalized exceedance
$\log(Y/w)$ is asymptotically independent of $X$ conditional on
$Y>w$. This gives a way to test homogeneity without estimating tail
indices.

To see why, for a high threshold $w$, define $Z:=Y/w$. Then, for
$r\geq1$, 
\[
\Pr(Z>r\mid X=x,Y>w)=r^{-\alpha(x)}\frac{L(rw;x)}{L(w;x)}.
\]
The factor $\kappa(x)$ cancels from this ratio. Under a uniform regular-variation
condition, the slowly varying ratio tends to one uniformly in $x$
for each fixed $r$ as $w\to\infty$. Hence, under $H_{0}$, 
\[
\Pr(Z>r\mid X=x,Y>w)\longrightarrow r^{-\alpha_{0}},\qquad r\geq1,
\]
which no longer depends on $x$. The same independence restriction
holds after taking logarithms.

We test this independence using distance covariance \citep{SzekelyRizzoBakirov2007},
which detects nonlinear, multivariate dependence, and calibrate by
permutation. The global test is consistent when its dependence signal
is large relative to sampling noise. Under fixed alternatives, however,
tail selection favors smaller tail indices \citep{Yang2026} and can
weaken this signal. Local balancing can protect power by selecting
extremes within two covariate groups before pooling. We also give
a sample-split adaptive version: an independent training subsample
learns whether an observation that already exceeds a high threshold
also exceeds a higher one, and a simple second-order probit score
uses this signal to form two approximately equal groups in the inference
subsample.

Existing tests of conditional tail-index homogeneity usually start
with a model for the conditional tail or extreme quantiles. They then
test a restriction within that model. For example, \citet{WangTsai2009}
model the Pareto index parametrically, \citet{WangLi2013} develop
a power-transformed quantile-regression framework, \citet{KinsvaterFried2017}
consider regression and rank-based procedures for conditional tail
behavior, and \citet{MomokiYoshida2024} study tests in structured
varying-coefficient tail-index models. A related literature examines
homogeneity across other settings. \citet{EinmahlDeHaanZhou2016}
allow the frequency of extremes to vary under a common extreme-value
index, \citet{deHaanZhou2021} test variation over a deterministic
design, \citet{EinmahlYangZhou2021} study constancy across angular
regions in multivariate regular variation, and \citet{ChenZhou2026}
study equality of extreme-value indices across many dimensions. These
approaches are useful when the form of the tail-index difference is
itself of interest. The radial-angular factorization used by \citet{WanDavis2019}
is mathematically related to our conditional-exceedance characterization,
but their objective is threshold selection.

Our contribution is to use conditional-exceedance independence as
an omnibus specification device without estimating tail indices. The
main theoretical difficulty is that independence holds only asymptotically
under the null. We account for finite-threshold dependence, establish
uniform validity that justifies order-statistic thresholds, and give
consistency and $K_{n}^{-1/2}$ local-power results for the global
and locally balanced tests. A short conditioning argument gives null
validity for a two-group partition learned on an independent training
subsample. The fixed-alternative limit also carries over when the
learned groups stabilize. This treatment connects to the approximate-randomization
literature, including \citet{ChungRomano2013}, \citet{CanayRomanoShaikh2017}
and \citet{CanayKamat2018}.

Simulations show competitive performance and gains from balancing
under strong tail concentration. In the insurance application, the
non-rejections by our test and the Wang--Li and Kinsvater--Fried
procedures motivate considering common-tail-index quantile methods,
although they do not establish homogeneity.

The remainder of the paper is organized as follows. Section~\ref{sec:model}
presents the model and the tail-independence characterization. Section~\ref{sec:test}
constructs the distance-covariance test. Section~\ref{sec:null-validity}
establishes permutation validity, and Section~\ref{sec:power} studies
power. Section~\ref{sec:simulations} reports the Monte Carlo experiments.
Section~\ref{sec:application} presents the insurance application,
and the final section concludes. Additional discussion and all technical
proofs are collected in the appendices.

\section{Model and Tail-Independence Characterization}

\label{sec:model}

Let $\{(Y_{i},X_{i}):1\leq i\leq n\}$ be independent copies of $(Y,X)$,
where $Y>0$. For the continuous-covariate power results below, $X\in\R^{d}$
has density $f_{X}$ with support $\mathcal{X}$. The null permutation
arguments also apply to discrete covariates, with sums replacing covariate
integrals. Suppose 
\begin{equation}
\bar{F}(y\mid x)=\kappa(x)y^{-\alpha(x)}L(y;x),\qquad y\geq y_{0},\label{eq:tail-model}
\end{equation}
with $\kappa(x)>0$, $\alpha(x)$ bounded away from zero and infinity
on the relevant support, and $L(\cdot;x)$ positive, slowly varying
in its first argument, and normalized so that $L(y;x)\to1$ for each
fixed $x$. We test 
\begin{align*}
H_{0} & :\ \alpha(x)=\alpha_{0}\quad P_{X}\text{-almost surely},\\
H_{1} & :\ \alpha(x)\ \text{is not }P_{X}\text{-almost surely constant}.
\end{align*}
The null allows $\kappa(x)$ to vary freely subject to regularity
conditions and permits covariate-dependent second-order behavior through
$L(y;x)$. For Pareto-type tails, the usual extreme-value index is
$\gamma(x)=1/\alpha(x)$, so homogeneity of $\alpha(x)$ and $\gamma(x)$
is equivalent.

We first develop the theory for a deterministic threshold $w_{n}\to\infty$
and define 
\[
p_{n}:=\Pp(Y>w_{n}),\qquad K_{n}:=np_{n},
\]
where $p_{n}\to0$ and $K_{n}\to\infty$. For this fixed-threshold
construction, $K_{n}$ is the expected number of exceedances. The
realized number is 
\[
\mathcal{I}_{n}:=\{i:Y_{i}>w_{n}\},\qquad N_{n}:=|\mathcal{I}_{n}|.
\]
For the order-statistic construction below, $K_{n}$ instead denotes
the chosen number of retained observations. With a fixed threshold,
the selected observations are iid conditional on $N_{n}$. This helps
in the main null calculation. Section~\ref{sec:null-validity} then
proves permutation validity for a range of increasing thresholds and
for the order-statistic threshold used in applications. For $i\in\mathcal{I}_{n}$,
define 
\[
Z_{i,n}:=Y_{i}/w_{n},\qquad Z_{i,n}^{*}:=\log Z_{i,n}.
\]
We use the log exceedance rather than the raw normalized exceedance.
The monotone transformation leaves the independence restriction unchanged,
but it replaces a limiting Pareto variable, whose first moment is
infinite when $\alpha_{0}\leq1$, by an $\Exp(\alpha_{0})$ variable
with moments of all orders.

At finite thresholds, departures from the limiting Pareto form can
create covariate dependence even under the null. The following uniform
Hall expansion, \citet{Hall1982}, uses $D(x)$ and $\beta(x)$ to
describe the magnitude and decay rate of these departures. Together
with the threshold-rate conditions below, it controls this bias so
that it does not cause asymptotic over-rejection.

\begin{assumption}[Uniform Hall expansion]\label{ass:hall} There
exist a constant $\underline{\beta}>0$, a measurable function $\beta(x)\geq\underline{\beta}$,
a bounded measurable function $D(x)$, and a remainder $r(y,x)$ such
that 
\[
L(y;x)=1+D(x)y^{-\beta(x)}+r(y,x),\qquad y\geq y_{0},
\]
with $\sup_{x\in\mathcal{X}}|D(x)|<\infty$. The remainder is differentiable
in $y$ for sufficiently large $y$ and satisfies 
\begin{equation}
\sup_{x\in\mathcal{X}}\sup_{u\geq y}u^{\beta(x)}\left\{ |r(u,x)|+u|\partial_{u}r(u,x)|\right\} \longrightarrow0.\label{eq:hall-remainder}
\end{equation}
\end{assumption}

Write 
\[
a_{n}:=w_{n}^{-\underline{\beta}},\qquad\lambda_{n}(x):=w_{n}^{-[\beta(x)-\underline{\beta}]}.
\]
Then $w_{n}^{-\beta(x)}=a_{n}\lambda_{n}(x)$. Because $L(y;x)$ is
normalized to converge to one, $\kappa(x)$ is the leading tail-frequency
coefficient and remains unrestricted across $x$. The lower bound
$\underline{\beta}$ provides a common envelope for the second-order
bias. No common value of $\beta(x)$ is imposed.

Under $H_{0}$, for each fixed $z\geq0$, 
\[
\Pp(Z_{i,n}^{*}>z\mid X_{i}=x,Y_{i}>w_{n})=e^{-\alpha_{0}z}\frac{L(w_{n}e^{z};x)}{L(w_{n};x)}.
\]
Applying Assumption~\ref{ass:hall} gives, uniformly in $x$, 
\begin{align}
 & \Pp(Z_{i,n}^{*}>z\mid X_{i}=x,Y_{i}>w_{n})\nonumber \\
 & \quad=e^{-\alpha_{0}z}\left[1+a_{n}\lambda_{n}(x)D(x)\{e^{-\beta(x)z}-1\}+o(a_{n})\right].\label{eq:hall-exceedance}
\end{align}
The uniform remainder follows from \eqref{eq:hall-remainder} and
$w_{n}^{-\beta(x)}\leq w_{n}^{-\underline{\beta}}=a_{n}$. Hence the
conditional distribution of $Z_{i,n}^{*}$ approaches an $\operatorname{Exp}(\alpha_{0})$
distribution that does not depend on $x$. Thus $Z_{i,n}^{*}$ and
$X_{i}$ are asymptotically independent within the selected tail sample.
We therefore use the independence restriction as a direct way to test
$H_{0}$ without estimating $\alpha(x)$.

\section{Test Construction}

\label{sec:test}

The procedure uses only the exceedance observations. Throughout this
section Euclidean distance is used on the support of $X$. We begin
with the deterministic threshold $w_{n}$. Because the original sample
is iid, conditional on $N_{n}$ the selected pairs $(Z_{j}^{*,e},X_{j}^{e})$,
$j=1,\ldots,N_{n}$, are iid from the joint density 
\[
f_{n}(x,z):=f_{X,n}(x)f_{n}(z\mid x),\qquad f_{X,n}(x):=\frac{f_{X}(x)\Pr(Y>w_{n}\mid X=x)}{\Pr(Y>w_{n})},
\]
where $f_{n}(z\mid x)$ is the density of $\log(Y/w_{n})$ conditional
on $X=x,Y>w_{n}$. Under $H_{0}$, \eqref{eq:hall-exceedance} makes
$f_{n}(z\mid x)$ approach the same exponential density for every
$x$. We exploit this implication directly with an independence test.
We first define the distance-covariance statistic and permutation
calibration, then describe the global and locally balanced ways of
selecting the exceedances. Sections~\ref{sec:null-validity} and~\ref{sec:power}
establish their null validity and power properties.

\subsection{Distance covariance statistic}

Let $D^{Z}$ and $D^{X}$ denote the pairwise distance matrices in
the selected tail sample: 
\[
D_{j\ell}^{Z}:=\left|Z_{j}^{*,e}-Z_{\ell}^{*,e}\right|,\qquad D_{j\ell}^{X}:=\left\Vert X_{j}^{e}-X_{\ell}^{e}\right\Vert ,\qquad j,\ell=1,\ldots,N_{n}.
\]
Let $H_{N}:=I_{N_{n}}-N_{n}^{-1}\bm{1}\bm{1}'$ and define the double-centered
distance matrices 
\[
A:=H_{N}D^{Z}H_{N},\qquad B:=H_{N}D^{X}H_{N}.
\]
The sample squared distance covariance is 
\begin{equation}
\widehat{\Delta}_{n}:=\widehat{\dCov}_{n}^{\,2}(Z^{*},X)=\frac{1}{N_{n}^{2}}\sum_{j=1}^{N_{n}}\sum_{\ell=1}^{N_{n}}A_{j\ell}B_{j\ell}=\frac{1}{N_{n}^{2}}\operatorname{tr}(AB).\label{eq:dcov}
\end{equation}
For the population analysis, let $\phi_{n}(s,t)$ be the joint characteristic
function of $(X,Z^{*})$ under $f_{n}$ and let $\phi_{X,n}(s)$ and
$\phi_{Z,n}(t)$ be its marginals. Define $c_{q}:=\pi^{(1+q)/2}/\Gamma((1+q)/2)$.
The population squared distance covariance is 
\begin{equation}
\Delta_{n}:=\frac{1}{c_{d}c_{1}}\int_{\R^{d}}\int_{\R}\frac{|\phi_{n}(s,t)-\phi_{X,n}(s)\phi_{Z,n}(t)|^{2}}{\|s\|^{d+1}|t|^{2}}\,dt\,ds.\label{eq:population-dcov}
\end{equation}
Under the usual finite-first-moment conditions, population distance
covariance is zero if and only if $Z^{*}$ and $X$ are independent
\citep{SzekelyRizzoBakirov2007}. Large values of \eqref{eq:dcov}
therefore provide evidence against tail-index homogeneity.

\subsection{Permutation calibration}

We calibrate the test by permuting the observed log exceedances across
the covariates in the selected tail sample. Let $\Pi_{N_{n}}$ denote
the set of permutations of $\{1,\ldots,N_{n}\}$. For $\pi\in\Pi_{N_{n}}$,
recompute \eqref{eq:dcov} from the pairs $(Z_{\pi(j)}^{*,e},X_{j}^{e})$,
and denote the resulting statistic by $\widehat{\Delta}_{n}^{\pi}$.
For the theoretical statements below, assume the permutation statistics
are almost surely distinct. This only avoids boundary randomization.
Let $c_{n,1-\eta}^{\pi}$ be the conditional $(1-\eta)$ quantile
of $\{\widehat{\Delta}_{n}^{\pi}:\pi\in\Pi_{N_{n}}\}$. The permutation
test rejects when 
\begin{equation}
\widehat{\Delta}_{n}>c_{n,1-\eta}^{\pi}.\label{eq:permutation-rejection}
\end{equation}

Under exact independence, permuting the log exceedances leaves the
joint null distribution invariant while preserving both marginal samples.
Section~\ref{sec:null-validity} justifies this calibration under
the approximate independence implied by the Hall expansion.

\subsection{Global order-statistic threshold}

For the global version, write $Y_{1:n}\leq\cdots\leq Y_{n:n}$ for
the order statistics. Given $K_{n}$ tail sample observations, set
\[
\widehat{w}_{n}:=Y_{n-K_{n}:n},
\]
and retain the $K_{n}$ observations above this threshold. The test
uses their log exceedances $Z_{i,n}^{*}:=\log(Y_{i}/\widehat{w}_{n})$
in \eqref{eq:dcov}.

For any two observations selected by the global threshold,

\begin{equation}
\left|\log\left(\frac{Y_{i}}{\widehat{w}_{n}}\right)-\log\left(\frac{Y_{j}}{\widehat{w}_{n}}\right)\right|=|\log Y_{i}-\log Y_{j}|.\label{eq:orderstat-distance-cancel}
\end{equation}
Thus, after selecting the upper $K_{n}$ observations, the numerical
value of the global threshold does not enter the pairwise distances.

\subsection{Locally balanced tail selection}

Under tail-index heterogeneity, a common high threshold increasingly
selects observations from covariate regions with smaller tail indices.
The resulting loss of covariate variation can weaken the global test,
see Section \ref{sec:power-design} for details. The purpose of local
balancing is therefore simple: split the sample into two covariate
groups that, under the alternative, tend to have different tail indices,
and retain upper observations within each group. Ideally, one group
contains observations with relatively larger $\alpha(x)$ and the
other observations with relatively smaller $\alpha(x)$. The goal
is not to estimate $\alpha(x)$ itself, but to preserve a contrast
in tail decay across the retained sample.

Let $\mathcal{X}_{n}:=\sigma(X_{1},\ldots,X_{n})$ and let $\mathcal{I}_{n,1}$
and $\mathcal{I}_{n,2}$ be an $\mathcal{X}_{n}$-measurable partition
of the observation indices. Write $N_{n,j}:=|\mathcal{I}_{n,j}|$
for $j=1,2$. Allocate the total number of retained observations $K_{n}$
approximately equally across the two groups, with $K_{n,1}+K_{n,2}=K_{n}$.
Within group $j$, write the corresponding outcomes in increasing
order as 
\[
Y_{j,1:N_{n,j}}\leq\cdots\leq Y_{j,N_{n,j}:N_{n,j}},
\]
and define the local threshold 
\[
\widehat{w}_{n,j}:=Y_{j,N_{n,j}-K_{n,j}:N_{n,j}}.
\]
Thus exactly the largest $K_{n,j}$ observations in group $j$ are
retained. For a retained observation $i$ in group $j$, define 
\begin{equation}
Z_{i,n}^{\mathrm{loc}}:=\log\left(\frac{Y_{i}}{\widehat{w}_{n,j}}\right).\label{eq:local-window-log-exceedance}
\end{equation}
The retained pairs from the two groups are pooled, and the dCov statistic
is calibrated by globally permuting the pooled $Z_{i,n}^{\mathrm{loc}}$
values across the retained covariates. Section~\ref{sec:adaptive-localization}
explains how we construct the two groups without estimating conditional
tail indices.

\subsection{Learning the two groups}

\label{sec:adaptive-localization}

We use an independent training subsample to learn a score that ranks
covariates by relative tail persistence. Randomly split the observations
into training and inference subsamples, $\mathcal{S}_{n}^{\mathrm{tr}}$
and $\mathcal{S}_{n}^{\mathrm{inf}}$, with $|\mathcal{S}_{n}^{\mathrm{tr}}|/n\to\rho\in(0,1)$.
Let $m_{n}:=|\mathcal{S}_{n}^{\mathrm{tr}}|$ and choose integers
$M_{1n}>M_{2n}$. From the training outcomes, define two order-statistic
thresholds 
\[
\widehat{w}_{1n}:=Y_{m_{n}-M_{1n}:m_{n}}^{\mathrm{tr}},\qquad\widehat{w}_{2n}:=Y_{m_{n}-M_{2n}:m_{n}}^{\mathrm{tr}},\qquad\widehat{w}_{1n}<\widehat{w}_{2n}.
\]
Among the $M_{1n}$ training observations above $\widehat{w}_{1n}$,
set 
\[
D_{i}:=\ind\{Y_{i}>\widehat{w}_{2n}\}.
\]
For deterministic $w_{1}<w_{2}$, 
\[
\Pr(Y>w_{2}\mid Y>w_{1},X=x)=\left(\frac{w_{2}}{w_{1}}\right)^{-\alpha(x)}\frac{L(w_{2};x)}{L(w_{1};x)}.
\]
The leading tail-frequency factor $\kappa(x)$ therefore cancels.
Since $w_{2}/w_{1}>1$, a larger conditional persistence probability
corresponds, up to the slowly varying term, to a smaller tail index
and hence a heavier tail.

We fit a probit model for $D_{i}$ using a generic second-order feature
vector $q(X_{i})$. For $x=(x_{1},\ldots,x_{d})'$, $q(x)$ contains
all linear terms, all squares, and all pairwise interactions, 
\[
q(x)=\bigl(x_{1},\ldots,x_{d},x_{1}^{2},\ldots,x_{d}^{2},\{x_{r}x_{s}:r<s\}\bigr)^{\prime}.
\]
Writing $\Phi$ for the standard normal distribution function, the
working model is 
\[
\Pr(D_{i}=1\mid X_{i}=x,\ Y_{i}>\widehat{w}_{1n})\approx\Phi\{\beta_{0}+q(x)'\beta\}.
\]
This probit is only a ranking device. It is not a model or estimator
of $\alpha(x)$.

Apply the fitted score $\widehat{p}_{n}(X_{i})$ to the inference
subsample and split observations at its median. The lower-score half
is intended to contain relatively larger $\alpha(x)$ values, while
the upper-score half is intended to contain relatively smaller $\alpha(x)$
values. We then apply the locally balanced construction above, allocating
the retained tail observations approximately equally across the two
groups. Conditional on the training subsample, the fitted score is
fixed and the two groups are measurable functions of the inference
covariates. Hence correct specification of the probit model is not
required for size validity. Section~\ref{sec:null-validity} gives
the conditioning argument.

\subsection{Implementation and tuning}

\label{sec:implementation-tuning}

Our default adaptive implementation uses a training fraction $\rho=1/5$.
Within the training sample, we set 
\[
M_{1n}=\lfloor0.30m_{n}\rfloor,\qquad M_{2n}=\lfloor0.25M_{1n}\rfloor,
\]
so the two training thresholds correspond to approximately $30\%$
and $7.5\%$ of the training sample. The larger separation between
the two training thresholds provides a clearer tail-persistence signal
while retaining enough observations for the probit fit.

The probit uses the generic second-order features described above.
Thus, for a scalar covariate the regressors are $(X,X^{2})$, while
for two covariates they are 
\[
X_{1},\quad X_{2},\quad X_{1}^{2},\quad X_{2}^{2},\quad X_{1}X_{2}.
\]
The inference sample is split at the median fitted score, producing
two approximately equal-sized groups, and $K_{n}$ is allocated as
evenly as possible across them. Ties in fitted scores are broken deterministically
using the covariates and then the original observation index. Outcome
ties in the local orderings are broken by the original observation
index. Apart from the usual choice of $K_{n}$, these tuning choices
are kept fixed in our implementation. The global test remains a useful
benchmark because it uses the same dCov statistic without local balancing.

\section{Null Validity}

\label{sec:null-validity}

\subsection{Hall drift and equivalence to independence}

The Hall expansion describes the dependence that remains at a finite
threshold under $H_{0}$. We impose the following condition and show
the result in the proposition followed. 

\begin{assumption}[Tail-design moments]\label{ass:dcov-moments}
For some $\delta>0$, 
\[
\sup_{n}\Ee\left[\|X\|^{1+\delta}\mid Y>w_{n}\right]<\infty.
\]
\end{assumption}

\begin{proposition}[Second-order distance-covariance drift] \label{prop:dcov-drift}
Suppose $H_{0}$, Assumptions~\ref{ass:hall} and \ref{ass:dcov-moments}
hold. Then 
\begin{equation}
\begin{aligned}\Delta_{n} & =O(a_{n}^{2})=O(w_{n}^{-2\underline{\beta}}),\\
\dCov(X,Z^{*}\mid Y>w_{n}) & =O(a_{n})=O(w_{n}^{-\underline{\beta}}).
\end{aligned}
\label{eq:dcov-equivalence}
\end{equation}
\end{proposition}

Since $K_{n}\widehat{\Delta}_{n}$ is the natural scaling under independence
\citep{SzekelyRizzoBakirov2007}, the Hall drift is negligible if

\begin{equation}
\sqrt{K_{n}}a_{n}=\sqrt{K_{n}}w_{n}^{-\underline{\beta}}\longrightarrow0.\label{eq:bias-rate}
\end{equation}
Condition \eqref{eq:bias-rate} uses only the lower bound $\underline{\beta}$,
not the unknown function $\beta(x)$. It may therefore be stronger
than needed. Similar second-order conditions appear in extreme-value
theory. For example, \citet{Hoga2017} imposes a condition of the
form $\sqrt{k}A(n/k)\to0$, while \citet{deHaanZhou2021} use its
uniform analogue for their local/global extreme-value-index theory.

\subsection{Uniform and order-statistic threshold validity}

Under exact independence, permuting the log exceedances leaves the
joint null distribution invariant. We first establish uniform size
validity for deterministic thresholds and for thresholds chosen using
the outcomes. Write $\bar{F}(y):=\Pp(Y>y)$.

\begin{proposition}[Permutation validity for global thresholds]
\label{prop:permutation} Assume $H_{0}$, Assumption~\ref{ass:hall},
and $0<\Ee[\kappa(X)]<\infty$. Let $\mathcal{W}_{n}:=[\underline{w}_{n},\overline{w}_{n}]$
be a deterministic threshold range such that 
\begin{equation}
\underline{w}_{n}\longrightarrow\infty,\qquad n\Pp(Y>\overline{w}_{n})\longrightarrow\infty,\label{eq:uniform-threshold-range}
\end{equation}
and 
\begin{equation}
\sqrt{n\Pp(Y>\underline{w}_{n})}\,\underline{w}_{n}^{-\underline{\beta}}\longrightarrow0.\label{eq:uniform-threshold-rate}
\end{equation}
For each deterministic $w\in\mathcal{W}_{n}$, let $\phi_{n,w}$ denote
the permutation test in \eqref{eq:permutation-rejection} constructed
from the observations satisfying $Y_{i}>w$. Then 
\begin{equation}
\sup_{w\in\mathcal{W}_{n}}\left|\Pp_{H_{0}}\{\phi_{n,w}=1\}-\eta\right|\longrightarrow0.\label{eq:uniform-permutation-size}
\end{equation}

If $Y$ has a continuous upper tail, the same holds for any threshold
$\widetilde{w}_{n}$ measurable with respect to $(Y_{1},\ldots,Y_{n})$
whenever 
\begin{equation}
\Pp\{\widetilde{w}_{n}\in\mathcal{W}_{n}\}\longrightarrow1.\label{eq:ymeasurable-threshold}
\end{equation}
In particular, let $K_{n}\to\infty$, $K_{n}/n\to0$, and choose $u_{n}$
so that $\bar{F}(u_{n})=K_{n}/n$. If $\sqrt{K_{n}}\,u_{n}^{-\underline{\beta}}\to0$,
the test at $\widehat{w}_{n}=Y_{n-K_{n}:n}$ has asymptotic level
$\eta$. \end{proposition}

The lower endpoint controls the worst-case Hall bias, while the upper
endpoint must retain a diverging number of exceedances. The order-statistic
proof uses a deterministic range around $u_{n}$. See Appendix~\ref{app:proofs}.

\subsection{Locally balanced permutation validity}

We now consider the two-group locally balanced construction of Section~\ref{sec:test},
including the dependence created by its two random group-specific
order-statistic thresholds.

\begin{proposition}[Locally balanced permutation validity] \label{prop:local-window-permutation}
Suppose $H_{0}$ and Assumption~\ref{ass:hall} hold, the two groups
are covariate-measurable, and the conditional distribution of $Y$
given $X=x$ is continuous in the upper tail uniformly over $x$.
If 
\begin{equation}
\min_{j=1,2}K_{n,j}\longrightarrow\infty,\qquad\max_{j=1,2}\frac{K_{n,j}}{N_{n,j}}\longrightarrow_{p}0,\label{eq:local-window-intermediate}
\end{equation}
and 
\begin{equation}
\sum_{j=1}^{2}K_{n,j}\widehat{w}_{n,j}^{-2\underline{\beta}}\longrightarrow_{p}0,\label{eq:local-window-bias-rate}
\end{equation}
then the global permutation test based on the pooled locally normalized
exceedances in \eqref{eq:local-window-log-exceedance} has asymptotic
level $\eta$. \end{proposition}

Condition \eqref{eq:local-window-bias-rate} implies that all local
thresholds diverge: since $\min_{j=1,2}K_{n,j}\to\infty$, 
\begin{equation}
\min_{j=1,2}\widehat{w}_{n,j}\longrightarrow_{p}\infty.\label{eq:local-window-diverging-thresholds}
\end{equation}

If additionally $N_{n,j}/n\to_{p}p_{j}\in(0,1)$ for $j=1,2$ and
the two local intermediate order statistics admit deterministic equivalents
$u_{n,j}$ satisfying $\widehat{w}_{n,j}/u_{n,j}\to_{p}1$ simultaneously,
then 
\[
\sqrt{K_{n}}\max_{j=1,2}u_{n,j}^{-\underline{\beta}}\longrightarrow0
\]
is sufficient for \eqref{eq:local-window-bias-rate}.

Unlike the single-threshold identity in \eqref{eq:orderstat-distance-cancel},
different local thresholds do not cancel from pairwise log-exceedance
distances across the two groups. The proof of Proposition~\ref{prop:local-window-permutation}
therefore conditions on the two group-specific boundaries and selected
identities.

The same validity result applies to the sample-split adaptive construction
in Section~\ref{sec:adaptive-localization}. Conditional on the random
split and the training subsample, the learned score is fixed, while
the inference observations remain iid. The two score groups are measurable
functions of the inference covariates, so Proposition~\ref{prop:local-window-permutation}
applies conditionally whenever its conditions hold with probability
tending to one. The conditional rejection probability therefore converges
to $\eta$, and bounded convergence gives the same unconditional limit.
Correct specification of the training probit is not required for this
null-validity argument.

\section{Power and Tail-Design Concentration}

\label{sec:power}

\begin{assumption}[Tail-sample moments]\label{ass:dcov-process-moments}
For some $\delta>0$, along any sequence for which we invoke a sample
dCov limit, 
\begin{equation}
\sup_{n}\Ee\!\left[(\|X\|+Z^{*})^{2+\delta}\mid Y>w_{n}\right]<\infty.\label{eq:process-moment}
\end{equation}
\end{assumption}

Under an alternative, write 
\[
\Delta_{n}^{A}:=\dCov^{2}(X,Z^{*}\mid Y>w_{n}).
\]
For the locally balanced test, let $\mathcal{G}_{n}$ contain the
retained covariates, group memberships, thresholds, and selected identities,
and let $U_{n}$ be uniform on $\{1,\ldots,K_{n}\}$ conditional on
$\mathcal{G}_{n}$. Define 
\begin{equation}
\Delta_{n}^{\mathrm{LB},A}:=\dCov^{2}(X_{U_{n},n}^{\mathrm{LB}},Z_{U_{n},n}^{\mathrm{LB}}\mid\mathcal{G}_{n}).\label{eq:local-balanced-dcov-signal}
\end{equation}
We first study how global tail selection changes the covariate design
and then give consistency and local-power results.

\subsection{Tail concentration and the effect of local balancing}

\label{sec:power-design}

At a fixed alternative, Bayes' rule gives the covariate density among
observations exceeding $w$ as 
\[
f_{X\mid Y>w}(x)\propto\kappa(x)w^{-\alpha(x)}L(w;x)f_{X}(x).
\]
By Assumption~\ref{ass:hall}, $L(w;x)\to1$ uniformly in $x$, so
for large $w$ this density is approximately proportional to $\kappa(x)w^{-\alpha(x)}f_{X}(x)$.
Thus increasing $w$ favors regions with smaller $\alpha(x)$. For
the two-group construction, let $\mathcal{C}_{1}$ and $\mathcal{C}_{2}$
be fixed groups of positive probability, let $\mathcal{X}_{j}$ be
the support of $X\mid X\in\mathcal{C}_{j}$, and define 
\[
\alpha_{\min}:=\inf_{x\in\mathcal{X}}\alpha(x),\qquad\alpha_{j}^{*}:=\inf_{x\in\mathcal{X}_{j}}\alpha(x),\quad j=1,2,
\]
and write the corresponding minimizer sets as 
\[
\mathcal{M}:=\{x\in\mathcal{X}:\alpha(x)=\alpha_{\min}\},\qquad\mathcal{M}_{j}:=\{x\in\mathcal{X}_{j}:\alpha(x)=\alpha_{j}^{*}\},\quad j=1,2.
\]
Thus $\alpha_{j}^{*}$ is the smallest tail exponent accessible in
group $j$.

\begin{assumption}[Fixed alternatives]\label{ass:fixed-alt} $\alpha$
and $\kappa$ are continuous on $\mathcal{X}$, $\alpha$ is bounded
above and away from zero, $\kappa>0$, and $0<\Ee[\kappa(X)]<\infty$.
Assumptions~\ref{ass:hall} and~\ref{ass:dcov-moments} hold uniformly,
also after conditioning on either group. \end{assumption}

\begin{theorem}[Tail concentration and the effect of local balancing]\label{thm:degeneracy}
Under Assumption~\ref{ass:fixed-alt}, as $w_{n}\to\infty$, $\alpha(X)\to\alpha_{\min}$
conditional on $Y>w_{n}$ and 
\begin{equation}
\Delta_{n}^{A}\longrightarrow0.\label{eq:global-signal-vanishes}
\end{equation}
For the two fixed groups, suppose the distribution of $X$ given $X\in\mathcal{C}_{j},Y>w$
converges to that of $X_{j}^{\infty}$, the groupwise moment bound
holds uniformly, and \eqref{eq:local-window-intermediate} and \eqref{eq:local-window-diverging-thresholds}
hold. If $K_{n,j}/K_{n}\to\rho_{j}>0$, $\rho_{1}+\rho_{2}=1$, define
$(X_{\infty},Z_{\infty})$ by choosing group $j$ with probability
$\rho_{j}$, drawing $X_{\infty}\sim X_{j}^{\infty}$, and then drawing
$Z_{\infty}\sim\Exp(\alpha_{j}^{*})$ independently conditional on
$j$. Then 
\begin{equation}
\Delta_{n}^{\mathrm{LB},A}\longrightarrow_{p}\Delta_{\infty}^{\mathrm{LB}}:=\dCov^{2}(X_{\infty},Z_{\infty}).\label{eq:local-nondegenerate-limit}
\end{equation}
Moreover, $X_{j}^{\infty}$ is supported on $\mathcal{M}_{j}$ for
$j=1,2$. Hence, if $\alpha_{1}^{*}\neq\alpha_{2}^{*}$, then $\mathcal{M}_{1}\cap\mathcal{M}_{2}=\varnothing$
and the limit is nonzero. \end{theorem}

\begin{remark}[Global signal decay and consistency]\label{rem:degeneracy-rate}
There is no universal decay rate. Appendix~\ref{app:degeneracy-rate}
illustrates the logarithmic rates $(\log w_{n})^{-3}$ and $(\log w_{n})^{-5/2}$
in regular boundary and interior benchmarks. In such cases, $K_{n}$
growing faster than the corresponding power of $\log w_{n}$ satisfies
the signal condition for global consistency. \end{remark}

\begin{remark}[Dependence on the two-group split]\label{rem:balancing-limits}
If $\alpha_{1}^{*}=\alpha_{2}^{*}$, then $\Delta_{\infty}^{\mathrm{LB}}=0$.
For the adaptive version, the same limit applies when the learned
groups converge in symmetric difference to two fixed groups satisfying
the conditions of Theorem~\ref{thm:degeneracy}. Hence the split
matters for power, but not for null validity. \end{remark}

\subsection{Consistency}

For the global test impose Assumption~\ref{ass:dcov-process-moments}.
For the locally balanced test also assume upper-tail continuity, Assumption~\ref{ass:hall},
\eqref{eq:local-window-intermediate}, \eqref{eq:local-window-diverging-thresholds},
and, for some $\delta>0$, 
\begin{equation}
\max_{j=1,2}\frac{1}{K_{n,j}}\sum_{i\in S_{n,j}}\Ee_{A}\!\left[\bigl(\|X_{i}\|+Z_{i,n}^{\mathrm{loc}}\bigr)^{2+\delta}\mid\mathcal{G}_{n}\right]=O_{p}(1),\label{eq:local-balanced-process-moment}
\end{equation}
where $S_{n,j}$ is the retained set in group $j$.

\begin{theorem}[Consistency of both tests]\label{thm:global-power}
Under the respective conditions above, 
\begin{align}
K_{n}\Delta_{n}^{A}\to\infty & \quad\Longrightarrow\quad\Pr_{A}\{\widehat{\Delta}_{n}>c_{n,1-\eta}^{\pi}\}\to1,\label{eq:fixed-consistency}\\
K_{n}\Delta_{n}^{\mathrm{LB},A}\to_{p}\infty & \quad\Longrightarrow\quad\Pr_{A}\{\widehat{\Delta}_{n}^{\mathrm{LB}}>c_{n,1-\eta}^{\mathrm{LB},\pi}\}\to1.\label{eq:local-balanced-consistency}
\end{align}
\end{theorem}

Thus global consistency can coexist with tail-design concentration,
while the two-group test is consistent whenever the limit in Theorem~\ref{thm:degeneracy}
is positive. Corollary~\ref{cor:global-orderstat-power} in Appendix~\ref{app:global-orderstat-power}
covers the order-statistic threshold used in practice.

\subsection{Local-alternative limits}

For local alternatives, suppose 
\begin{equation}
\alpha_{n}(x)=\alpha_{0}+\delta_{n}g(x),\qquad\delta_{n}=c/\sqrt{K_{n}},\label{eq:local-alt}
\end{equation}
where $g$ is bounded, $c\in\R$ is fixed, and $\kappa$ and $L$
are unchanged. Let $\Pr_{n,c}$ denote probability under the sequence
in \eqref{eq:local-alt}. Under $H_{0}$, the limiting covariate density
in the global tail sample is 
\[
f_{X}^{0}(x):=\frac{\kappa(x)f_{X}(x)}{\Ee[\kappa(X)]}.
\]
For the locally balanced test, let $f_{X}^{\mathrm{LB}}$ denote the
limiting density of the retained covariates under the conditions in
Appendix~\ref{app:local-power}. For a covariate density $f$, let
$G_{f}$ be the centered Gaussian characteristic-function process
obtained under independent $X\sim f$ and $Z\sim\Exp(\alpha_{0})$,
and let $H_{g,f}$ be the deterministic first-order shift defined
in Appendix~\ref{app:local-power}. Write 
\[
T_{f}(c):=\|G_{f}+cH_{g,f}\|_{\mathrm{dCov}}^{2},\qquad T_{G}(c):=T_{f_{X}^{0}}(c),\qquad T_{\mathrm{LB}}(c):=T_{f_{X}^{\mathrm{LB}}}(c).
\]
For $a\in\{G,\mathrm{LB}\}$, let $q_{a,1-\eta}$ denote the $(1-\eta)$
quantile of $T_{a}(0)$.

\begin{theorem}[Local-alternative limits]\label{thm:local-power}
Under \eqref{eq:local-alt} and the respective conditions in Appendix~\ref{app:local-power},
\begin{align}
K_{n}\widehat{\Delta}_{n} & \Longrightarrow T_{G}(c),\label{eq:sample-local-power}\\
K_{n}\widehat{\Delta}_{n}^{\mathrm{LB}} & \Longrightarrow T_{\mathrm{LB}}(c).\label{eq:local-balanced-local-power-stat}
\end{align}
If $q_{a,1-\eta}$ is unique and is a continuity point of the distributions
of both $T_{a}(0)$ and $T_{a}(c)$, then the corresponding permutation
test satisfies 
\begin{equation}
\Pr_{n,c}(\mathrm{reject})\longrightarrow\Pr\{T_{a}(c)>q_{a,1-\eta}\},\qquad a\in\{G,\mathrm{LB}\}.\label{eq:local-power-function}
\end{equation}
\end{theorem}

At $c=0$, $T_{f}(0)=\|G_{f}\|_{\mathrm{dCov}}^{2}$, so Theorem~\ref{thm:local-power}
gives the sampling null limit of the scaled dCov statistic. Under
the continuity condition above, \eqref{eq:local-power-function} reduces
to the nominal rejection probability $\eta$. For fixed $c\neq0$,
the same theorem gives the local asymptotic power against $K_{n}^{-1/2}$
departures. These distributional limits complement the size results
in Section~\ref{sec:null-validity}. Implementation still uses permutation
critical values and does not require estimating $q_{a,1-\eta}$.

\section{Simulations}

\label{sec:simulations}

\subsection{Data-generating processes and threshold choice}

The five designs examine size and different forms of tail-index heterogeneity.
DGP~1 checks size in the presence of finite-threshold second-order
bias. DGP~2 studies scalar local alternatives. DGP~3a has smooth
bivariate heterogeneity. DGP~3b is a pure interaction. DGP~3c
adds nonlinear curvature that is not contained in the working score
model. Throughout, $U\sim\operatorname{Unif}(0,1)$ is independent
of the covariates.

The proposed procedure is the adaptive two-group construction in Sections~\ref{sec:adaptive-localization}
and~\ref{sec:implementation-tuning}. We randomly use one fifth of
the sample for training and the remaining four fifths for inference.
In the training sample, $M_{1}=\lfloor0.30m\rfloor$ and $M_{2}=\lfloor0.25M_{1}\rfloor$.
The inference observations are divided at the median fitted score,
and the total $k$ retained tail observations are allocated approximately
equally across the two groups. For scalar covariates the probit score
uses $(X,X^{2})$. For bivariate covariates it uses the same generic
second-order basis in every design, 
\[
(X_{1},X_{2},X_{1}^{2},X_{2}^{2},X_{1}X_{2}).
\]
Thus the score specification is not changed across DGPs~3a--3c.

We use sample sizes $n\in\{2000,5000\}$ and two tail sample sizes,
\begin{equation}
\mathcal{K}_{2000}:=\{100,50\},\qquad\mathcal{K}_{5000}:=\{200,100\}.\label{eq:simulation-k-grid}
\end{equation}
We write $k_{L}$ and $k_{S}$ for the larger and smaller values,
respectively.

\paragraph{DGP 1: Hall-type null with $\beta=1$.}

Let $X\sim\operatorname{Unif}[-1,1]$ and $s(x):=\exp(x/2)$. Generate
\[
Y:=s(X)\left(U^{-1/2}-1\right).
\]
Thus the tail exponent is constant at $\alpha_{0}=2$, while the leading
tail-frequency term is $\kappa(x)=s^{2}(x)=e^{x}$. The Hall second-order
exponent is $\beta=1$, with coefficient $D(x)=-2s(x)$.

\paragraph{DGP 2: Linear scalar local alternatives.}

Let $X\sim\operatorname{Unif}[-1,1]$, define $\kappa(x):=\exp(x/2)$
and $g(x):=\sqrt{3}\,x$. For each $(n,k)$ in \eqref{eq:simulation-k-grid},
set 
\begin{equation}
\alpha_{n,c,k}(x):=2+\frac{c}{\sqrt{k}}g(x)=2+\frac{c\sqrt{3}}{\sqrt{k}}x,\label{eq:dgp2-alpha}
\end{equation}
and generate 
\[
Y:=\left\{ \frac{\kappa(X)}{U}\right\} ^{1/\alpha_{n,c,k}(X)}.
\]
We use $c\in\{0,1,2,4\}$. The case $c=0$ is an exact-Pareto null
with covariate-dependent tail frequency. The case $c>0$ gives a $k^{-1/2}$
local departure. This design is deliberately favorable to regression-based
competitors.

\paragraph{DGP 3a: Smooth bivariate heterogeneity.}

Let $X_{1},X_{2}\stackrel{\mathrm{iid}}{\sim}\operatorname{Unif}[-1,1]$
and 
\[
\alpha_{\theta}(X):=2\exp\{\theta X_{1}\},\qquad\kappa(X):=\exp(X_{2}/2),
\]
with 
\[
Y:=\left\{ \frac{\kappa(X)}{U}\right\} ^{1/\alpha_{\theta}(X)}.
\]
Here $X_{1}$ affects tail thickness and $X_{2}$ affects tail frequency
only. We use $\theta\in\{0,0.15,0.30,0.45\}$.

\paragraph{DGP 3b: Pure interaction.}

Let $X_{1},X_{2}\stackrel{\mathrm{iid}}{\sim}\operatorname{Unif}[-1,1]$
and 
\[
\alpha_{\theta}(X):=2\exp\{\theta X_{1}X_{2}\},\qquad\kappa(X):=1,\qquad Y:=U^{-1/\alpha_{\theta}(X)}.
\]
We use $\theta\in\{0,0.5,1,1.5\}$. Neither $X_{1}$ nor $X_{2}$
alone has a first-order linear projection on the interaction. The
generic second-order adaptive score does contain $X_{1}X_{2}$, but
the same score basis is used in all bivariate designs rather than
chosen for this DGP.

\paragraph{DGP 3c: Cubic-perturbed nonlinear heterogeneity.}

Again let $X_{1},X_{2}\stackrel{\mathrm{iid}}{\sim}\operatorname{Unif}[-1,1]$.
Define 
\[
h_{0}(x):=x_{1}x_{2}+\frac{1}{2}(x_{2}^{2}-x_{1}^{2}),
\]
and, with $\delta=0.5$, 
\[
g_{\delta}(x):=\frac{h_{0}(x)+\frac{\delta}{2}(x_{1}^{3}+x_{2}^{3})}{1+\delta}.
\]
Set 
\[
\alpha_{\theta}(X):=2\exp\{\theta g_{\delta}(X)\},\qquad\kappa(X):=1,\qquad Y:=U^{-1/\alpha_{\theta}(X)},
\]
for $\theta\in\{0,0.5,1,1.5\}$. The quadratic part rotates the interaction
structure, while the cubic terms are deliberately omitted from the
generic second-order probit score. DGP~3c therefore checks the adaptive
procedure under working-model misspecification. At the same time,
the cubic terms have nonzero linear projections under the uniform
design, so the linear regression competitors are not made mechanically
powerless.

\subsection{Competing procedures and implementation}

We report four procedures: the adaptive two-group dCov test, the global
top-$k$ dCov benchmark, the Kinsvater--Fried L-test, and the Wang--Li
test. The adaptive test uses the sample split and tuning described
above. The global test uses the same dCov statistic and the top $k$
observations from the full sample. The two regression competitors
use the same nominal $k$ but otherwise follow their own constructions.
For all bivariate designs, their working regressions contain the linear
main effects $X_{1}$ and $X_{2}$. We do not add design-specific
interactions or polynomial terms.

\paragraph{Wang--Li three-stage test.}

We implement the constancy test of \citet{WangLi2013} with $\lambda=0$
(the log transformation). Transformed linear quantile regressions
produce covariate-dependent EVI estimates $\widehat{\gamma}(X_{i})$,
and 
\[
T_{n}^{WL}:=\frac{1}{n}\sum_{i=1}^{n}\{\widehat{\gamma}(X_{i})-\widehat{\gamma}_{P}\}^{2},\qquad\widehat{\gamma}_{P}:=\frac{1}{n}\sum_{i=1}^{n}\widehat{\gamma}(X_{i}).
\]
With the nonconstant regressors centered, we use $S_{n}^{WL}=kT_{n}^{WL}/\widehat{\gamma}_{P}^{2}$
and the $\chi_{d}^{2}$ reference distribution from Corollary~3.1
of \citet{WangLi2013}: one degree of freedom for scalar designs and
two for bivariate designs.

\paragraph{Kinsvater--Fried L-test.}

We use the L-estimator-based test of \citet{KinsvaterFried2017}.
For a given $k$, their method estimates a high conditional quantile,
forms approximately $k$ relative excesses, and estimates a linear
conditional EVI model. Following their recommendation, we use 20 probabilities
equally spaced over $[1/2,1-1/40]$ and their optimal L-weights. Constancy
is tested by the joint Wald test that the slope coefficients are zero.

All procedures are evaluated at the 5\% nominal level with 1,000 Monte
Carlo replications. Both dCov tests use 299 random permutations with
the plus-one correction. The adaptive and global dCov procedures use
the same nominal total tail count $k$.

\subsection{Simulation results}

Tables~\ref{tab:sim-dgp12}--\ref{tab:sim-dgp3c} in Appendix~\ref{app:simulation-additional}
report the rejection frequencies. Figures~\ref{fig:sim-size}--\ref{fig:sim-dgp3c}
summarize them with pointwise 95\% Wilson intervals.

\paragraph{Size and scalar local power.}

Under DGP~1, adaptive dCov, Kinsvater--Fried, and Wang--Li are
close to nominal size. Global dCov shows mild over-rejection in one
design.

\begin{figure}[H]
\centering \includegraphics[width=0.96\textwidth]{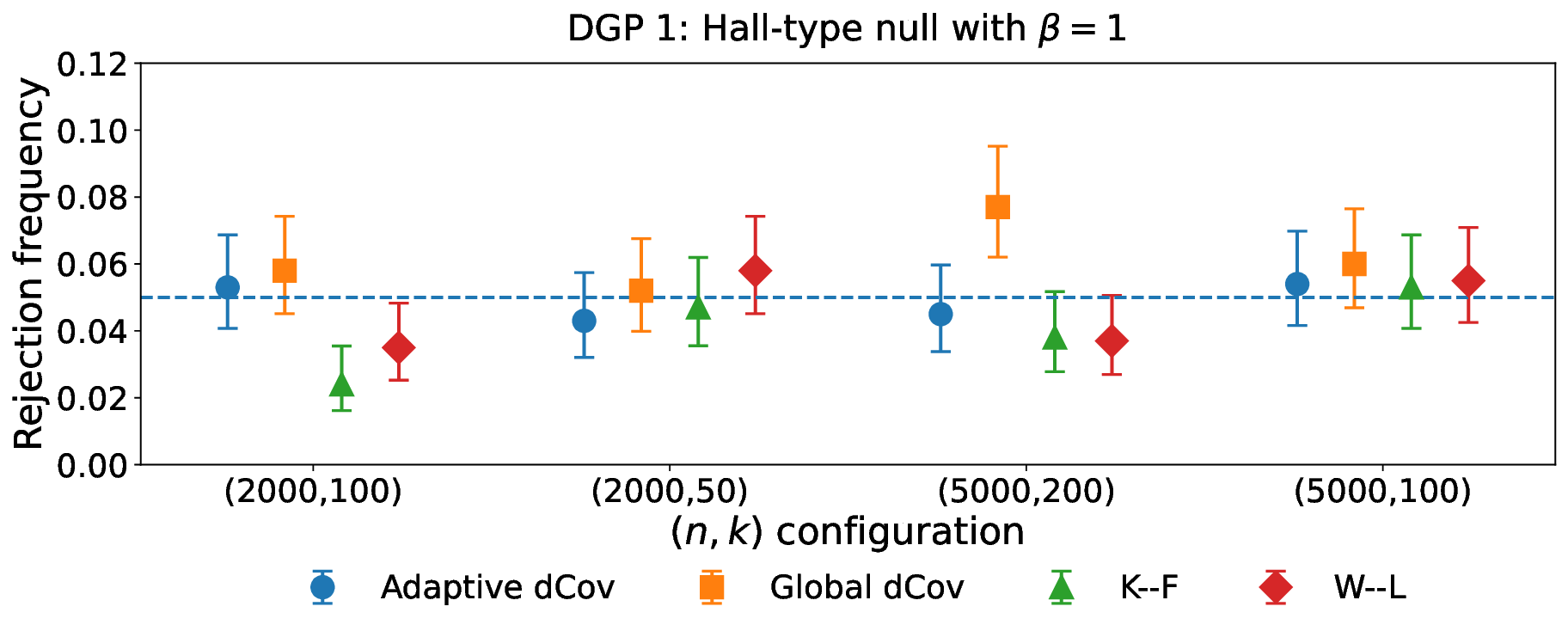}

\begin{singlespace}
\caption{Size under the Hall-type null (DGP~1).}
\label{fig:sim-size} 
\end{singlespace}
\end{figure}

Under DGP~2, the four procedures have similar power for moderate
local departures, while the regression tests are somewhat stronger
at the largest $c$, as expected from the locally linear design.

\begin{figure}[H]
\centering \includegraphics[width=0.96\textwidth]{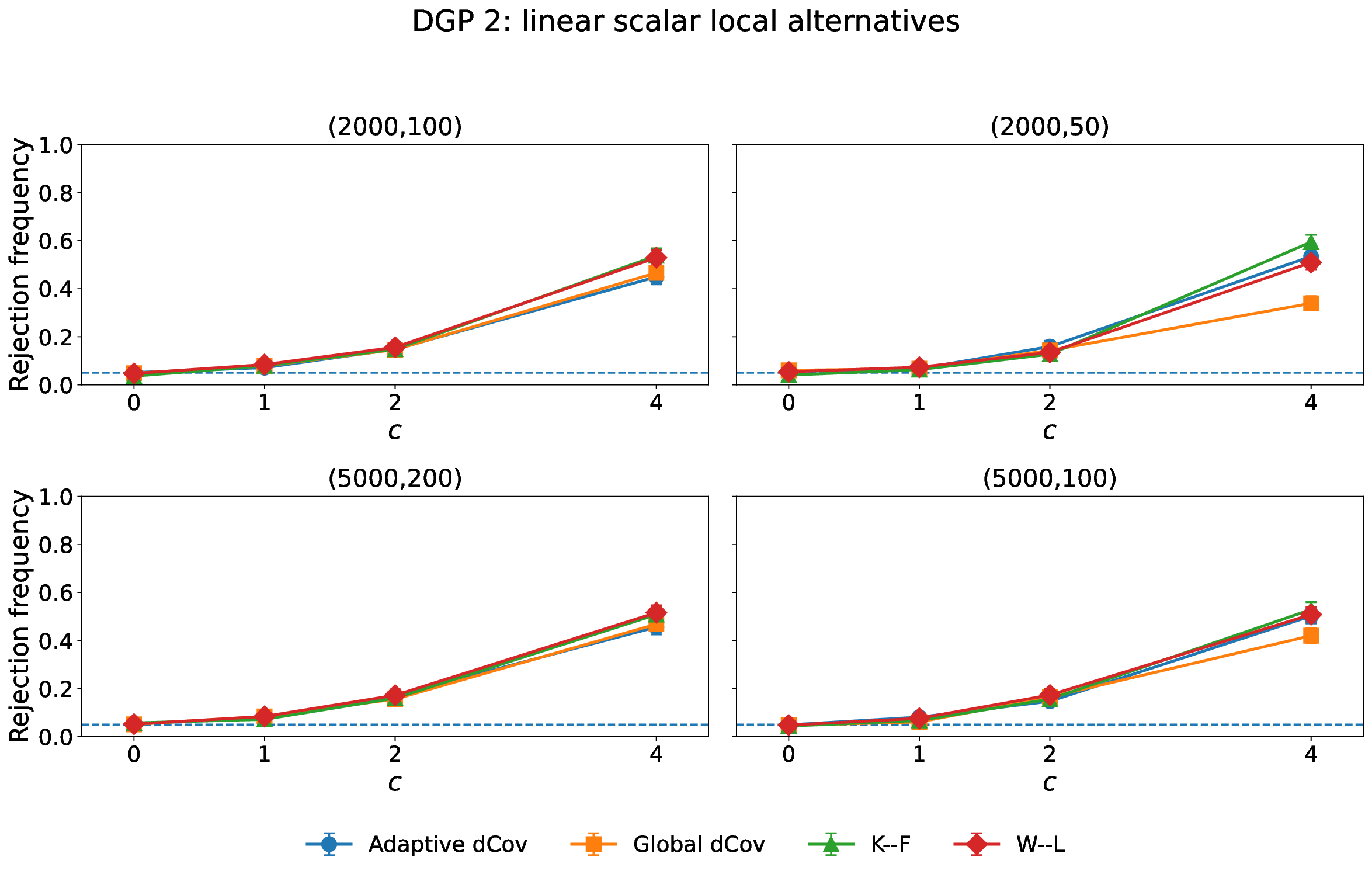}

\begin{singlespace}
\caption{Power under linear scalar local alternatives (DGP~2).}
\label{fig:sim-dgp2} 
\end{singlespace}
\end{figure}

\paragraph{Smooth and interaction alternatives.}

Under DGP~3a, the regression procedures benefit from the smooth main
effect, although adaptive dCov gains substantially over global dCov
as heterogeneity strengthens.

\begin{figure}[H]
\centering \includegraphics[width=0.96\textwidth]{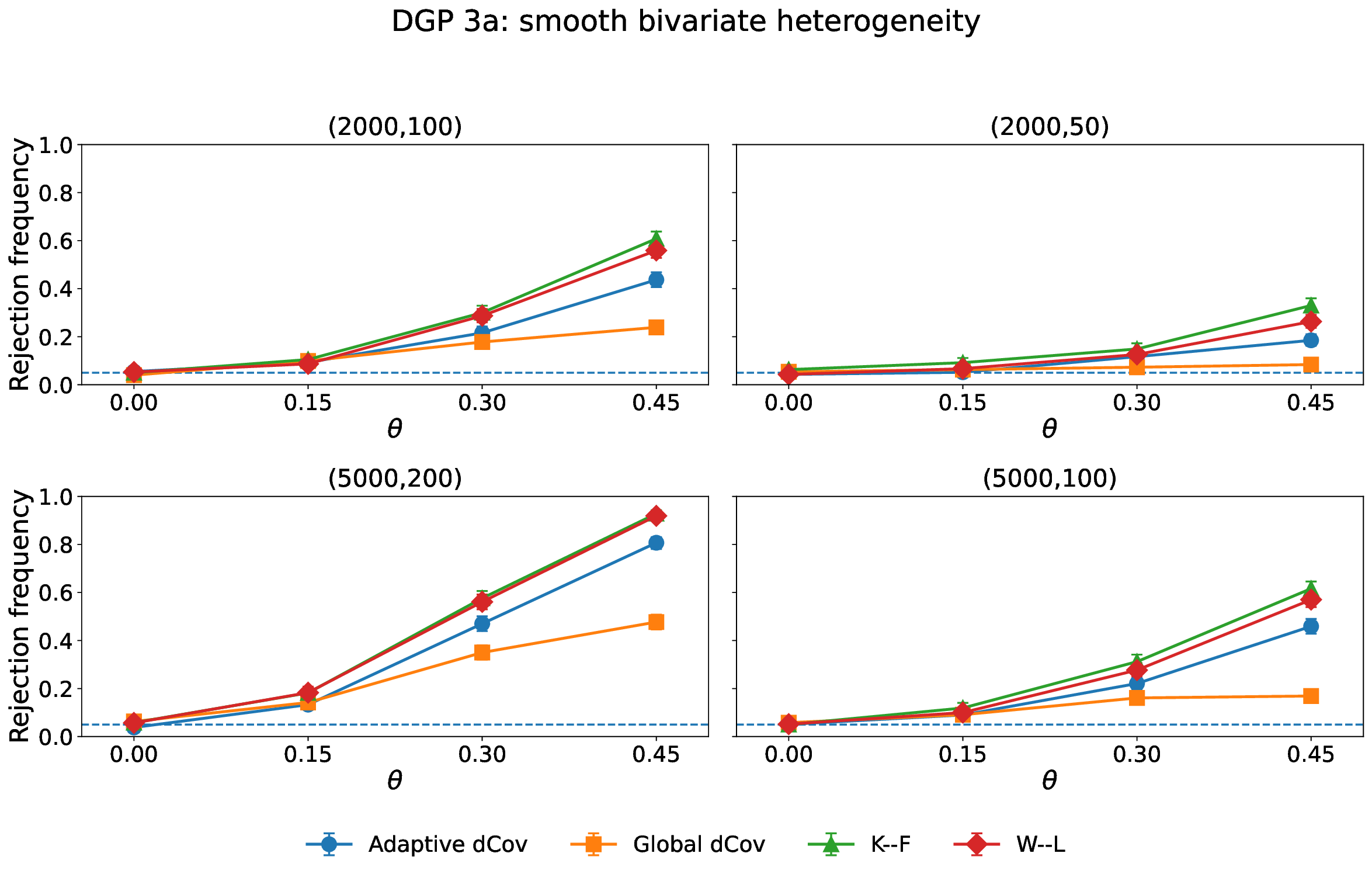}

\begin{singlespace}
\caption{Power under smooth bivariate heterogeneity (DGP~3a).}
\label{fig:sim-dgp3a} 
\end{singlespace}
\end{figure}

Under the pure interaction in DGP~3b, adaptive dCov becomes the most
powerful as the departure strengthens, while the main-effect Kinsvater--Fried
regression has little power. This illustrates the value of learning
a useful localization when the relevant functional form is not specified.

\begin{figure}[H]
\centering \includegraphics[width=0.96\textwidth]{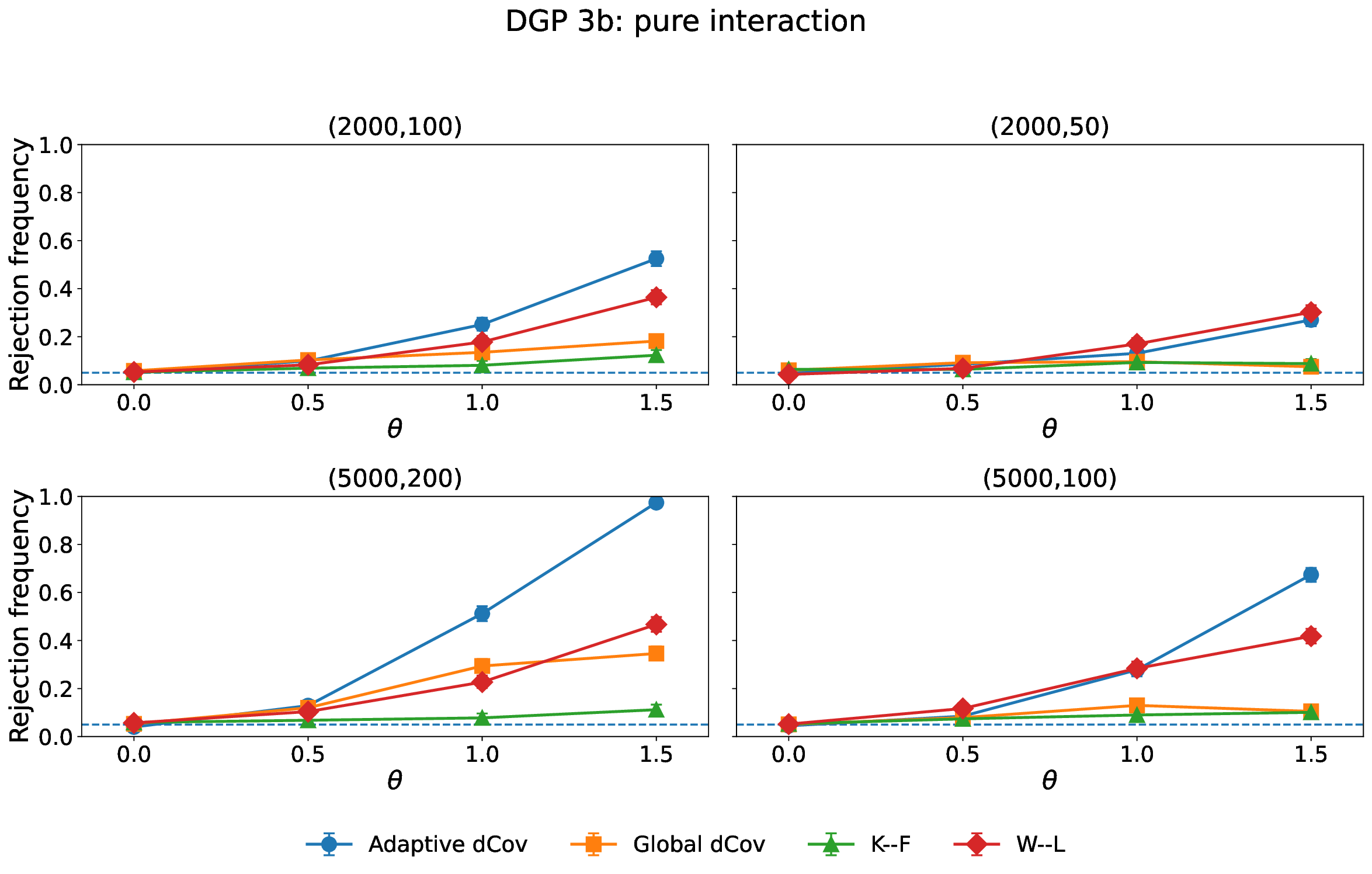}

\begin{singlespace}
\caption{Power under a pure nonlinear interaction (DGP~3b).}
\label{fig:sim-dgp3b} 
\end{singlespace}
\end{figure}

\paragraph{Misspecified nonlinear alternative.}

In DGP~3c, where the adaptive score omits the cubic terms, all procedures
remain correctly sized and adaptive dCov remains competitive at the
stronger alternatives, without uniformly dominating Wang--Li.

\begin{figure}[H]
\centering \includegraphics[width=0.96\textwidth]{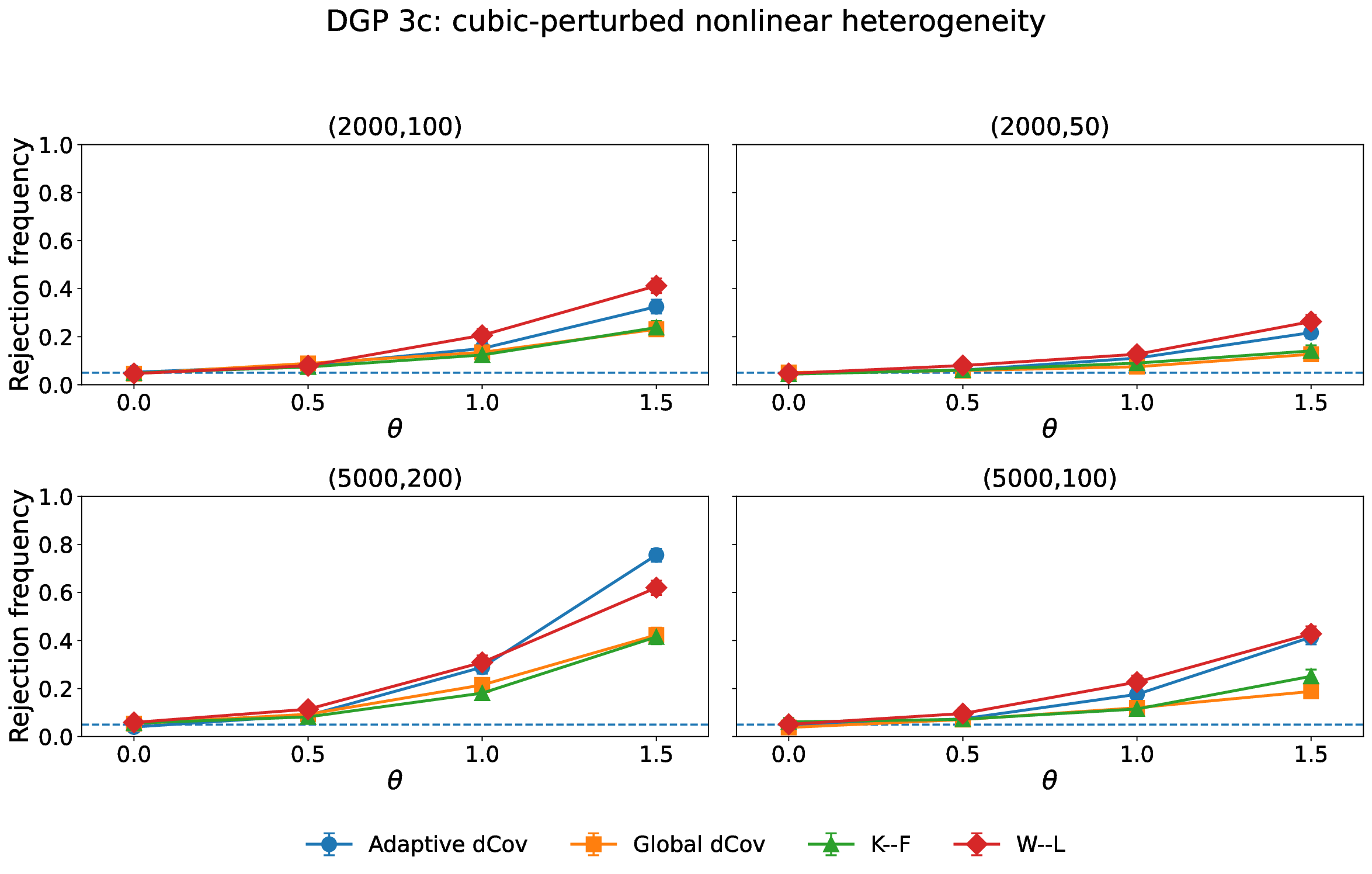}

\begin{singlespace}
\caption{Power under cubic-perturbed nonlinear heterogeneity (DGP~3c).}
\label{fig:sim-dgp3c} 
\end{singlespace}
\end{figure}

\paragraph{Effect of local balancing.}

Table~\ref{tab:global-local} shows that global and adaptive dCov
are similar near the null, but adaptive balancing can be substantially
more powerful under stronger interaction and nonlinear alternatives.
This matches the tail-design concentration mechanism in Section~\ref{sec:power}.

\FloatBarrier

\section{Empirical Illustration}

\label{sec:application}

We use the French motor third-party liability data \texttt{freMTPL2}
from the \texttt{CASdatasets} package.\footnote{See \url{https://dutangc.github.io/CASdatasets/reference/freMTPL.html}.}
After aggregating claims within policy and retaining positive claim
amounts, the sample contains 24,944 observations. We study standardized
\texttt{BonusMalus} and \texttt{VehPower}.

\subsection{Tail-index homogeneity}

We use $k\in\{1247,624\}$, approximately 5\% and 2.5\% of the usable
sample. Adaptive dCov uses the same 20\% training split, generic second-order
score, and two-group tuning as in the simulations, with 999 permutations.
We repeat the random split ten times and summarize the resulting $p$-values.
The global benchmark uses the top $k$ observations, and the Kinsvater--Fried
and Wang--Li tests use the same nominal tail counts. Table~\ref{tab:application-tail-tests}
reports the results.

\begin{table}[!htbp]
\centering \caption{Tail-index homogeneity tests for French motor claim severity}
\label{tab:application-tail-tests} {\small{}{}}{\small{}%
\begin{tabular}{rlrl}
\hline 
{\small$k$  } & {\small Method  } & {\small Statistic  } & {\small$p$-value or 10-split summary }\tabularnewline
\hline 
{\small 1247  } & {\small Adaptive dCov  } & {\small ---  } & {\small min 0.129, max 0.710, median 0.458, mean 0.455 }\tabularnewline
 & {\small Global dCov  } & {\small 0.000941  } & {\small 0.720 }\tabularnewline
 & {\small K--F  } & {\small 2.051  } & {\small 0.359 }\tabularnewline
 & {\small W--L  } & {\small 3.940  } & {\small 0.139 }\tabularnewline
{\small 624  } & {\small Adaptive dCov  } & {\small ---  } & {\small min 0.596, max 0.996, median 0.825, mean 0.832 }\tabularnewline
 & {\small Global dCov  } & {\small 0.001932  } & {\small 0.730 }\tabularnewline
 & {\small K--F  } & {\small 1.603  } & {\small 0.449 }\tabularnewline
 & {\small W--L  } & {\small 2.895  } & {\small 0.235 }\tabularnewline
\hline 
\end{tabular}}{\small\par}

{\small\smallskip{}
 \raggedright}{\footnotesize\textit{Notes:}}{\footnotesize{} The
covariates are standardized BonusMalus and VehPower. Adaptive dCov
uses a 20\% training subsample, the two-group score split described
in Section~\ref{sec:adaptive-localization}, and 999 permutations.
Global dCov uses the top $k$ observations from the full usable sample
and 999 permutations. For K--F the reported statistic is the Wald
statistic. For W--L it is the $\chi^{2}$ statistic. K--F denotes
the Kinsvater--Fried L-test and W--L denotes the Wang--Li test. }{\footnotesize\par}
\end{table}

None of the four procedures rejects tail-index homogeneity. Across
the ten adaptive splits, the $p$-values range from 0.129 to 0.710
for $k=1247$ and from 0.596 to 0.996 for $k=624$. The global and
regression-based tests lead to the same conclusion. The common EVI
estimates from Wang--Li and Kinsvater--Fried are also similar (about
0.81--0.85). Thus the data provide no evidence against a common tail
index for these covariates, although non-rejection does not establish
homogeneity.

\subsection{High-quantile regression}

For comparison, we estimate linear quantile regressions at $\tau\in\{0.95,0.975\}$,
\[
Q_{\tau}(Y\mid X)=\beta_{0,\tau}+\beta_{1,\tau}\text{BonusMalus}+\beta_{2,\tau}\text{VehPower},
\]
with 199 pairs-bootstrap replications. See Table~\ref{tab:application-qr}.

\begin{table}[!htbp]
\centering \caption{Upper conditional quantile regressions for claim severity}
\label{tab:application-qr} {\small{}{}{}{}{}{}{}{}{}{}{}{}{}{}{}{}{}{}{}{}{}{}{}{}{}{}{}}{\small{}%
\begin{tabular}{lrrrr}
\hline 
 & \multicolumn{2}{c}{{\small$\tau=0.95$}} & \multicolumn{2}{c}{{\small$\tau=0.975$}}\tabularnewline
{\small Variable  } & {\small Estimate  } & {\small 95\% CI  } & {\small Estimate  } & {\small 95\% CI }\tabularnewline
\hline 
{\small Intercept  } & {\small 5060.98  } & {\small$[4988.10,\ 5379.59]$  } & {\small 8868.06  } & {\small$[8307.30,\ 9232.81]$ }\tabularnewline
{\small BonusMalus  } & {\small 306.75  } & {\small$[193.58,\ 643.46]$  } & {\small 775.90  } & {\small$[61.61,\ 1442.23]$ }\tabularnewline
{\small VehPower  } & {\small 239.43  } & {\small$[7.74,\ 590.75]$  } & {\small 339.83  } & {\small$[-343.85,\ 847.52]$ }\tabularnewline
\hline 
\end{tabular}}{\small\par}

{\small\smallskip{}
 }{\small\par}

\raggedright{}{\footnotesize\textit{Notes:}}{\footnotesize{} The
dependent variable is policy-level claim severity. BonusMalus and
VehPower are standardized before estimation. Confidence intervals
are percentile intervals from 199 pairs-bootstrap replications. }{\footnotesize\par}

\end{table}

BonusMalus is positively associated with both upper conditional quantiles,
while vehicle power is significant only at $\tau=0.95$. Thus upper
quantiles vary with the covariates even though tail-index homogeneity
is not rejected. This supports considering common-tail-index extremal
quantile methods such as \citet{Chernozhukov2005} and \citet{WangLiHe2012},
subject to their remaining assumptions.

\section{Conclusion}

We propose a permutation distance-covariance test of tail-index homogeneity
without estimating tail indices. We establish validity and power results
for global and locally balanced versions. Global consistency can coexist
with tail concentration, while the adaptive two-group construction
can preserve covariate variation in the selected tail sample. Simulations
show competitive performance across linear, interaction, and misspecified
nonlinear alternatives. In French motor insurance, the tests do not
reject homogeneity, while risk characteristics affect upper conditional
quantiles.

The practical tradeoff is between reducing finite-threshold bias,
learning a useful two-group localization, and retaining enough tail
observations for inference. Null validity does not require correct
specification of the training score. The local-power analysis for
learned groups currently uses a stability condition. Relaxing that
condition is a natural direction for future research. \medskip{}

 \bibliographystyle{ecta}
\bibliography{references}

\newpage{}

\appendix
\setcounter{page}{1} 
\global\long\def\thepage{A\arabic{page}}%

\begin{center}
{\Large\textbf{Appendix}}{\Large\par}
\par\end{center}

\begin{center}
 
\par\end{center}

\section{Additional Simulation and Application Results}

\label{app:simulation-additional}

This appendix gives further results for the simulations and the insurance
application.

\subsection{Numerical simulation results}

Tables~\ref{tab:sim-dgp12}--\ref{tab:sim-dgp3c} give the design-specific
rejection frequencies used in Section~\ref{sec:simulations}.

\begin{table}[htbp]
\begin{singlespace}
\centering \caption{Empirical rejection frequencies: size and scalar local alternatives}
\label{tab:sim-dgp12} {\small{}{}}{\small{}%
\begin{tabular}{llcccc}
\hline 
 &  & \multicolumn{4}{c}{{\small$(n,k)$}}\tabularnewline
{\small Parameter  } & {\small Method  } & {\small$(2000,100)$  } & {\small$(2000,50)$  } & {\small$(5000,200)$  } & {\small$(5000,100)$ }\tabularnewline
\hline 
\multicolumn{6}{l}{{\small\textit{Panel A: DGP 1 (Hall-type null, $\beta=1$)}}}\tabularnewline
\hline 
{\small Null  } & {\small Adaptive dCov  } & {\small .053  } & {\small .043  } & {\small .045  } & {\small .054 }\tabularnewline
 & {\small Global dCov  } & {\small .058  } & {\small .052  } & {\small .077  } & {\small .060 }\tabularnewline
 & {\small K--F  } & {\small .024  } & {\small .047  } & {\small .038  } & {\small .053 }\tabularnewline
 & {\small W--L  } & {\small .035  } & {\small .058  } & {\small .037  } & {\small .055 }\tabularnewline
\multicolumn{6}{l}{{\small\textit{Panel B: DGP 2 (linear scalar local alternatives)}}}\tabularnewline
\hline 
{\small$c=0$ (Null)  } & {\small Adaptive dCov  } & {\small .051  } & {\small .055  } & {\small .052  } & {\small .049 }\tabularnewline
 & {\small Global dCov  } & {\small .048  } & {\small .060  } & {\small .051  } & {\small .047 }\tabularnewline
 & {\small K--F  } & {\small .036  } & {\small .040  } & {\small .057  } & {\small .044 }\tabularnewline
 & {\small W--L  } & {\small .047  } & {\small .054  } & {\small .052  } & {\small .048 }\tabularnewline
{\small$c=1$  } & {\small Adaptive dCov  } & {\small .071  } & {\small .067  } & {\small .080  } & {\small .081 }\tabularnewline
 & {\small Global dCov  } & {\small .078  } & {\small .067  } & {\small .084  } & {\small .061 }\tabularnewline
 & {\small K--F  } & {\small .079  } & {\small .063  } & {\small .072  } & {\small .067 }\tabularnewline
 & {\small W--L  } & {\small .084  } & {\small .073  } & {\small .084  } & {\small .076 }\tabularnewline
{\small$c=2$  } & {\small Adaptive dCov  } & {\small .148  } & {\small .159  } & {\small .171  } & {\small .147 }\tabularnewline
 & {\small Global dCov  } & {\small .146  } & {\small .143  } & {\small .156  } & {\small .165 }\tabularnewline
 & {\small K--F  } & {\small .148  } & {\small .127  } & {\small .160  } & {\small .156 }\tabularnewline
 & {\small W--L  } & {\small .156  } & {\small .135  } & {\small .172  } & {\small .173 }\tabularnewline
{\small$c=4$  } & {\small Adaptive dCov  } & {\small .449  } & {\small .534  } & {\small .456  } & {\small .502 }\tabularnewline
 & {\small Global dCov  } & {\small .466  } & {\small .339  } & {\small .468  } & {\small .420 }\tabularnewline
 & {\small K--F  } & {\small .537  } & {\small .594  } & {\small .508  } & {\small .529 }\tabularnewline
 & {\small W--L  } & {\small .529  } & {\small .509  } & {\small .516  } & {\small .508 }\tabularnewline
\hline 
\end{tabular}}{\small\par}

{\small\smallskip{}
 \raggedright}{\footnotesize\textit{Notes:}}{\footnotesize{} Rejection
frequencies are based on 1,000 Monte Carlo replications at the 5\%
level. Adaptive dCov uses a 20\% training sample, the two-group score
split in Section~\ref{sec:adaptive-localization}, and 299 permutations.
Global dCov uses the top $k$ observations and 299 permutations. K--F
denotes Kinsvater--Fried and W--L denotes Wang--Li. }{\footnotesize\par}
\end{singlespace}
\end{table}

\begin{table}[htbp]
\begin{singlespace}
\centering \caption{Empirical rejection frequencies: smooth and interaction alternatives}
\label{tab:sim-dgp3ab} {\small{}{}}{\small{}%
\begin{tabular}{llcccc}
\hline 
 &  & \multicolumn{4}{c}{{\small$(n,k)$}}\tabularnewline
{\small Parameter  } & {\small Method  } & {\small$(2000,100)$  } & {\small$(2000,50)$  } & {\small$(5000,200)$  } & {\small$(5000,100)$ }\tabularnewline
\hline 
\multicolumn{6}{l}{{\small\textit{Panel A: DGP 3a (smooth bivariate heterogeneity)}}}\tabularnewline
\hline 
{\small$\theta=0$ (Null)  } & {\small Adaptive dCov  } & {\small .055  } & {\small .043  } & {\small .038  } & {\small .051 }\tabularnewline
 & {\small Global dCov  } & {\small .040  } & {\small .054  } & {\small .063  } & {\small .058 }\tabularnewline
 & {\small K--F  } & {\small .050  } & {\small .063  } & {\small .058  } & {\small .050 }\tabularnewline
 & {\small W--L  } & {\small .053  } & {\small .043  } & {\small .057  } & {\small .052 }\tabularnewline
{\small$\theta=0.15$  } & {\small Adaptive dCov  } & {\small .093  } & {\small .052  } & {\small .133  } & {\small .090 }\tabularnewline
 & {\small Global dCov  } & {\small .099  } & {\small .063  } & {\small .142  } & {\small .091 }\tabularnewline
 & {\small K--F  } & {\small .105  } & {\small .092  } & {\small .182  } & {\small .119 }\tabularnewline
 & {\small W--L  } & {\small .087  } & {\small .067  } & {\small .183  } & {\small .100 }\tabularnewline
{\small$\theta=0.3$  } & {\small Adaptive dCov  } & {\small .216  } & {\small .117  } & {\small .470  } & {\small .221 }\tabularnewline
 & {\small Global dCov  } & {\small .178  } & {\small .073  } & {\small .350  } & {\small .161 }\tabularnewline
 & {\small K--F  } & {\small .300  } & {\small .149  } & {\small .576  } & {\small .312 }\tabularnewline
 & {\small W--L  } & {\small .287  } & {\small .126  } & {\small .561  } & {\small .277 }\tabularnewline
{\small$\theta=0.45$  } & {\small Adaptive dCov  } & {\small .437  } & {\small .185  } & {\small .807  } & {\small .459 }\tabularnewline
 & {\small Global dCov  } & {\small .239  } & {\small .084  } & {\small .477  } & {\small .169 }\tabularnewline
 & {\small K--F  } & {\small .608  } & {\small .330  } & {\small .928  } & {\small .616 }\tabularnewline
 & {\small W--L  } & {\small .559  } & {\small .263  } & {\small .919  } & {\small .570 }\tabularnewline
\multicolumn{6}{l}{{\small\textit{Panel B: DGP 3b (pure interaction)}}}\tabularnewline
\hline 
{\small$\theta=0$ (Null)  } & {\small Adaptive dCov  } & {\small .052  } & {\small .049  } & {\small .040  } & {\small .045 }\tabularnewline
 & {\small Global dCov  } & {\small .058  } & {\small .060  } & {\small .053  } & {\small .051 }\tabularnewline
 & {\small K--F  } & {\small .052  } & {\small .064  } & {\small .059  } & {\small .053 }\tabularnewline
 & {\small W--L  } & {\small .053  } & {\small .043  } & {\small .057  } & {\small .052 }\tabularnewline
{\small$\theta=0.5$  } & {\small Adaptive dCov  } & {\small .099  } & {\small .085  } & {\small .129  } & {\small .085 }\tabularnewline
 & {\small Global dCov  } & {\small .103  } & {\small .092  } & {\small .120  } & {\small .079 }\tabularnewline
 & {\small K--F  } & {\small .069  } & {\small .064  } & {\small .068  } & {\small .074 }\tabularnewline
 & {\small W--L  } & {\small .083  } & {\small .068  } & {\small .104  } & {\small .117 }\tabularnewline
{\small$\theta=1$  } & {\small Adaptive dCov  } & {\small .251  } & {\small .131  } & {\small .512  } & {\small .278 }\tabularnewline
 & {\small Global dCov  } & {\small .135  } & {\small .096  } & {\small .294  } & {\small .130 }\tabularnewline
 & {\small K--F  } & {\small .081  } & {\small .093  } & {\small .078  } & {\small .090 }\tabularnewline
 & {\small W--L  } & {\small .178  } & {\small .170  } & {\small .227  } & {\small .284 }\tabularnewline
{\small$\theta=1.5$  } & {\small Adaptive dCov  } & {\small .525  } & {\small .270  } & {\small .974  } & {\small .674 }\tabularnewline
 & {\small Global dCov  } & {\small .182  } & {\small .075  } & {\small .346  } & {\small .105 }\tabularnewline
 & {\small K--F  } & {\small .123  } & {\small .088  } & {\small .112  } & {\small .101 }\tabularnewline
 & {\small W--L  } & {\small .364  } & {\small .302  } & {\small .467  } & {\small .418 }\tabularnewline
\hline 
\end{tabular}}{\small\par}

{\small\smallskip{}
 \raggedright}{\footnotesize\textit{Notes:}}{\footnotesize{} Rejection
frequencies are based on 1,000 Monte Carlo replications at the 5\%
level. The adaptive score uses the generic second-order basis $(X_{1},X_{2},X_{1}^{2},X_{2}^{2},X_{1}X_{2})$
in both designs. The K--F and W--L working regressions use the linear
main effects $X_{1}$ and $X_{2}$ only. Both dCov procedures use
299 permutations. }{\footnotesize\par}
\end{singlespace}
\end{table}

\begin{table}[htbp]
\begin{singlespace}
\centering \caption{Empirical rejection frequencies: cubic-perturbed nonlinear heterogeneity
(DGP 3c)}
\label{tab:sim-dgp3c} {\small{}{}}{\small{}%
\begin{tabular}{llcccc}
\hline 
 &  & \multicolumn{4}{c}{{\small$(n,k)$}}\tabularnewline
{\small Parameter  } & {\small Method  } & {\small$(2000,100)$  } & {\small$(2000,50)$  } & {\small$(5000,200)$  } & {\small$(5000,100)$ }\tabularnewline
\hline 
\multicolumn{6}{l}{{\small\textit{DGP 3c}}}\tabularnewline
\hline 
{\small$\theta=0$ (Null)  } & {\small Adaptive dCov  } & {\small .052  } & {\small .049  } & {\small .040  } & {\small .045 }\tabularnewline
 & {\small Global dCov  } & {\small .046  } & {\small .051  } & {\small .054  } & {\small .038 }\tabularnewline
 & {\small K--F  } & {\small .049  } & {\small .044  } & {\small .056  } & {\small .061 }\tabularnewline
 & {\small W--L  } & {\small .047  } & {\small .047  } & {\small .059  } & {\small .051 }\tabularnewline
{\small$\theta=0.5$  } & {\small Adaptive dCov  } & {\small .081  } & {\small .061  } & {\small .088  } & {\small .074 }\tabularnewline
 & {\small Global dCov  } & {\small .089  } & {\small .059  } & {\small .093  } & {\small .071 }\tabularnewline
 & {\small K--F  } & {\small .074  } & {\small .061  } & {\small .082  } & {\small .072 }\tabularnewline
 & {\small W--L  } & {\small .079  } & {\small .080  } & {\small .114  } & {\small .096 }\tabularnewline
{\small$\theta=1$  } & {\small Adaptive dCov  } & {\small .151  } & {\small .111  } & {\small .289  } & {\small .176 }\tabularnewline
 & {\small Global dCov  } & {\small .135  } & {\small .075  } & {\small .215  } & {\small .120 }\tabularnewline
 & {\small K--F  } & {\small .124  } & {\small .090  } & {\small .181  } & {\small .115 }\tabularnewline
 & {\small W--L  } & {\small .206  } & {\small .127  } & {\small .309  } & {\small .227 }\tabularnewline
{\small$\theta=1.5$  } & {\small Adaptive dCov  } & {\small .325  } & {\small .217  } & {\small .756  } & {\small .414 }\tabularnewline
 & {\small Global dCov  } & {\small .231  } & {\small .127  } & {\small .423  } & {\small .188 }\tabularnewline
 & {\small K--F  } & {\small .238  } & {\small .141  } & {\small .415  } & {\small .251 }\tabularnewline
 & {\small W--L  } & {\small .412  } & {\small .263  } & {\small .620  } & {\small .428 }\tabularnewline
\hline 
\end{tabular}}{\small\par}

{\small\smallskip{}
 \raggedright}{\footnotesize\textit{Notes:}}{\footnotesize{} DGP~3c
uses $\delta=0.5$ and the same generic second-order adaptive score
as DGPs~3a--3b. The cubic terms in $g_{\delta}$ are not included
in that score. K--F and W--L use the linear main effects $X_{1}$
and $X_{2}$ only. K--F had at most two numerical failures among
1,000 replications in a design. Its reported rejection frequency uses
the successful replications. Both dCov procedures use 299 permutations. }{\footnotesize\par}
\end{singlespace}
\end{table}

\clearpage{}

\subsection{Global dCov benchmark}

Table~\ref{tab:global-local} reports the simple average across the
four $(n,k)$ configurations. Near the null, global and adaptive dCov
have similar rejection frequencies. The difference becomes larger
when fixed heterogeneity induces stronger concentration of the global
tail design.

\begin{table}[htbp]
\begin{singlespace}
\centering \caption{Comparison of global and adaptive two-group dCov}
\label{tab:global-local} {\small{}{}}{\small{}%
\begin{tabular}{llcc}
\hline 
{\small Design  } & {\small Parameter  } & {\small Global dCov  } & {\small Adaptive dCov }\tabularnewline
\hline 
{\small DGP 1  } & {\small Null  } & {\small .062  } & {\small .049 }\tabularnewline
{\small DGP 2  } & {\small$c=0$ (Null)  } & {\small .052  } & {\small .052 }\tabularnewline
{\small DGP 2  } & {\small$c=1$  } & {\small .073  } & {\small .075 }\tabularnewline
{\small DGP 2  } & {\small$c=2$  } & {\small .152  } & {\small .156 }\tabularnewline
{\small DGP 2  } & {\small$c=4$  } & {\small .423  } & {\small .485 }\tabularnewline
{\small DGP 3a  } & {\small$\theta=0$ (Null)  } & {\small .054  } & {\small .047 }\tabularnewline
{\small DGP 3a  } & {\small$\theta=0.15$  } & {\small .099  } & {\small .092 }\tabularnewline
{\small DGP 3a  } & {\small$\theta=0.3$  } & {\small .191  } & {\small .256 }\tabularnewline
{\small DGP 3a  } & {\small$\theta=0.45$  } & {\small .242  } & {\small .472 }\tabularnewline
{\small DGP 3b  } & {\small$\theta=0$ (Null)  } & {\small .056  } & {\small .047 }\tabularnewline
{\small DGP 3b  } & {\small$\theta=0.5$  } & {\small .099  } & {\small .100 }\tabularnewline
{\small DGP 3b  } & {\small$\theta=1$  } & {\small .164  } & {\small .293 }\tabularnewline
{\small DGP 3b  } & {\small$\theta=1.5$  } & {\small .177  } & {\small .611 }\tabularnewline
{\small DGP 3c  } & {\small$\theta=0$ (Null)  } & {\small .047  } & {\small .047 }\tabularnewline
{\small DGP 3c  } & {\small$\theta=0.5$  } & {\small .078  } & {\small .076 }\tabularnewline
{\small DGP 3c  } & {\small$\theta=1$  } & {\small .136  } & {\small .182 }\tabularnewline
{\small DGP 3c  } & {\small$\theta=1.5$  } & {\small .242  } & {\small .428 }\tabularnewline
\hline 
\end{tabular}}{\small\par}

{\small\smallskip{}
 \raggedright}{\footnotesize\textit{Notes:}}{\footnotesize{} Each
entry is the simple average of the four rejection frequencies for
$(n,k)=(2000,100),(2000,50),(5000,200),(5000,100)$. Both procedures
use the same nominal total tail count $k$. Adaptive dCov uses a 20\%
training split and two inference groups. }{\footnotesize\par}
\end{singlespace}
\end{table}

\clearpage{}

\subsection{Application details for global and adaptive dCov}

\label{app:application-dcov-details}

The adaptive test uses one fixed sample split for both tail counts.
With 24,944 usable observations, the training subsample contains 4,988
observations and the inference subsample 19,956. The median score
split therefore gives two inference groups of 9,978 observations each.
The training tail counts are $M_{1}=1,496$ and $M_{2}=374$.

Table~\ref{tab:application-dcov-details} reports the global threshold
and the two group-specific adaptive thresholds. For $k=1247$, the
adaptive allocation is 624 observations in the low-score group and
623 in the high-score group. For $k=624$, both groups contribute
312 observations. The lower local thresholds relative to the global
threshold reflect that the adaptive test selects the same total number
of tail observations from the 80\% inference subsample while balancing
them across the two learned groups.

\begin{table}[htbp]
\centering \caption{Application details for global and adaptive dCov}
\label{tab:application-dcov-details} {\small{}{}}{\small{}%
\begin{tabular}{rrrrrrrr}
\hline 
{\small$k$  } & {\small Global $\widehat{w}$  } & {\small$n_{\mathrm{low}}$  } & {\small$k_{\mathrm{low}}$  } & {\small$\widehat{w}_{\mathrm{low}}$  } & {\small$n_{\mathrm{high}}$  } & {\small$k_{\mathrm{high}}$  } & {\small$\widehat{w}_{\mathrm{high}}$ }\tabularnewline
\hline 
{\small 1247  } & {\small 5026.38  } & {\small 9978  } & {\small 624  } & {\small 4011.00  } & {\small 9978  } & {\small 623  } & {\small 4429.32 }\tabularnewline
{\small 624  } & {\small 8889.74  } & {\small 9978  } & {\small 312  } & {\small 7000.00  } & {\small 9978  } & {\small 312  } & {\small 8134.61 }\tabularnewline
\hline 
\end{tabular}}{\small\par}

{\small\smallskip{}
 \raggedright}{\footnotesize\textit{Notes:}}{\footnotesize{} ``Low''
and ``high'' refer to the fitted exceedance-persistence score. Higher
score corresponds approximately to smaller tail exponent. The score
is learned only on the training subsample and is held fixed when the
two tail counts are evaluated. }{\footnotesize\par}
\end{table}

\clearpage{}

\section{Proofs and Technical Results}

\label{app:proofs}

This appendix follows the order of the main results. Subsection~\ref{app:hall-null}
derives the Hall drift and null approximation, and Subsection~\ref{app:permutation-validity}
proves global and locally balanced permutation validity. Subsections~\ref{app:tail-design}
and~\ref{app:degeneracy-rate} study tail-design concentration and
the relation between global signal decay and consistency. Subsection~\ref{app:power-analysis}
contains the dCov sample calculation and the consistency proofs. Appendix~\ref{app:local-power}
gives the technical conditions and proof for the local-power theorem.

For any joint distribution $P$ of $(X,Z)$, write 
\[
\mathscr{C}_{P}(s,t):=\phi_{P}(s,t)-\phi_{X,P}(s)\phi_{Z,P}(t),
\]
and for any set $B\subset X,$
\[
P_{X,n}(B):=\int_{B}f_{X,n}(x)\,dx.
\]

\subsection{Hall expansion and null approximation}

\label{app:hall-null}

Differentiating \eqref{eq:hall-exceedance}, using \eqref{eq:hall-remainder},
yields the conditional density expansion 
\begin{equation}
\frac{f_{n}(z\mid x)}{f_{0}(z)}=1+a_{n}q_{n}(x,z)+o(a_{n}),\qquad z\geq0,\label{eq:density-expansion}
\end{equation}
uniformly in $x$ in the weighted $L^{2}$ sense used below, where
\begin{align*}
f_{0}(z) & :=\alpha_{0}e^{-\alpha_{0}z},\\
q_{n}(x,z) & :=\lambda_{n}(x)D(x)\psi_{\beta(x)}(z),\\
\psi_{b}(z) & :=\frac{\alpha_{0}+b}{\alpha_{0}}e^{-bz}-1,\qquad b>0.
\end{align*}
For every $b>0$, $\Ee_{f_{0}}[\psi_{b}(Z)]=0$.

The same expansion can be expressed through characteristic functions.
For a nonnegative random variable $V$, $\Ee(e^{itV})=1+it\int_{0}^{\infty}e^{itz}\Pp(V>z)\,dz$.
Hence \eqref{eq:hall-exceedance} gives 
\begin{equation}
\phi_{Z\mid X=x,n}(t)=\phi_{0}(t)+a_{n}g_{n}(x,t)+o(a_{n}),\label{eq:conditional-cf}
\end{equation}
where 
\begin{align*}
\phi_{0}(t) & :=\frac{\alpha_{0}}{\alpha_{0}-it},\\
g_{n}(x,t) & :=\lambda_{n}(x)D(x)b_{\beta(x)}(t),\\
b_{b}(t) & :=it\left\{ \frac{1}{\alpha_{0}+b-it}-\frac{1}{\alpha_{0}-it}\right\} ,\qquad b>0.
\end{align*}
The remainder in \eqref{eq:conditional-cf} is understood in the weighted
$L^{2}$ sense induced by \eqref{eq:population-dcov}. Assumption~\ref{ass:hall},
including the derivative bound in \eqref{eq:hall-remainder}, is sufficient
for this step.

Define 
\[
\mathcal{K}_{n}:=\frac{1}{c_{d}c_{1}}\int_{\R^{d}}\int_{\R}\frac{\left|\operatorname{Cov}_{P_{X,n}}\left(e^{is'X},g_{n}(X,t)\right)\right|^{2}}{\|s\|^{d+1}|t|^{2}}\,dt\,ds.
\]
\textbf{Proof of Proposition~\ref{prop:dcov-drift}.}

\noindent\begin{proof}Using iterated expectations and \eqref{eq:conditional-cf},
\begin{align}
 & \phi_{n}(s,t)-\phi_{X,n}(s)\phi_{Z,n}(t)\nonumber \\
 & \quad=a_{n}\operatorname{Cov}_{P_{X,n}}\left(e^{is'X},g_{n}(X,t)\right)+o(a_{n})\label{eq:cf-factorization}
\end{align}
in the weighted $L^{2}$ norm defining distance covariance. Substituting
\eqref{eq:cf-factorization} into \eqref{eq:population-dcov} gives
the sharper expansion 
\begin{equation}
\Delta_{n}=a_{n}^{2}\mathcal{K}_{n}+o(a_{n}^{2}).\label{eq:dcov-sharp-expansion}
\end{equation}

We now show that $\mathcal{K}_{n}$ is uniformly bounded. Since $|D(x)|$
is uniformly bounded, $0<\lambda_{n}(x)\leq1$, and 
\[
\frac{|b_{b}(t)|^{2}}{|t|^{2}}=\frac{b^{2}}{\{(\alpha_{0}+b)^{2}+t^{2}\}(\alpha_{0}^{2}+t^{2})}\leq\frac{1}{\alpha_{0}^{2}+t^{2}},\qquad b>0,
\]
the $t$-part is uniformly integrable over all possible values of
$\beta(x)$. Also, $|e^{is'X}-1|\leq\min\{2,\|s\|\|X\|\}$, and the
uniform tail-design moment condition controls the $s$-integral both
near zero and at infinity. Hence $\sup_{n}\mathcal{K}_{n}<\infty$.
Equation \eqref{eq:dcov-sharp-expansion} then implies both bounds
in \eqref{eq:dcov-equivalence}. \end{proof}

\subsection{Permutation validity}

\label{app:permutation-validity}

For densities $p,q$, define their total-variation distance and chi-square
divergence by 
\[
\|p-q\|_{\mathrm{TV}}:=\frac{1}{2}\int|p(u)-q(u)|\,du,\qquad\chi^{2}(p\|q):=\int\left(\frac{p(u)}{q(u)}-1\right)^{2}q(u)\,du.
\]
The divergence is $+\infty$ if $q$ vanishes on a set with positive
probability under $p$. For discrete observations, replace the integrals
by sums. We also write $\chi^{2}(P\|Q)$ for the corresponding distributions.
The bound 
\[
\|p-q\|_{\mathrm{TV}}\leq\sqrt{\chi^{2}(p\|q)/2}
\]
will be used below. Total variation bounds the difference in probabilities
of any rejection event, so approximation by an independent reference
experiment transfers its permutation validity to the actual sample.

Define the leading tail-design density for a constant exponent by
\begin{equation}
f_{X}^{0}(x):=\frac{\kappa(x)f_{X}(x)}{\int\kappa(v)f_{X}(v)\,dv},\qquad0<\int\kappa(v)f_{X}(v)\,dv<\infty.\label{eq:limiting-tail-design}
\end{equation}

\noindent\textbf{Proof of Proposition~\ref{prop:permutation}.}

\begin{proof}Let 
\[
\mu_{\kappa}:=\Ee[\kappa(X)]
\]
and use the limiting density $f_{X}^{0}$ in \eqref{eq:limiting-tail-design}.
Under $H_{0}$, differentiating the conditional survival function
in the upper tail gives 
\[
f_{Y\mid X}(y\mid x)=\kappa(x)y^{-\alpha_{0}-1}B(y,x),\qquad B(y,x):=\alpha_{0}L(y;x)-y\partial_{y}L(y;x).
\]
Assumption~\ref{ass:hall}, including its derivative remainder condition,
implies uniformly in $x$ that 
\[
B(y,x)=\alpha_{0}\{1+O(y^{-\underline{\beta}})\}.
\]
The common envelope $y^{-\underline{\beta}}$ also covers variation
of $\beta(x)$. Bayes' formula now gives an ordinary density ratio:
\begin{equation}
\frac{f_{X\mid Y=y}(x)}{f_{X}^{0}(x)}=\frac{B(y,x)}{\int B(y,v)f_{X}^{0}(v)\,dv}=1+\rho_{y}(x),\label{eq:uniform-threshold-mark-rn}
\end{equation}
Here $\rho_{y}(x):=f_{X\mid Y=y}(x)/f_{X}^{0}(x)-1$, and 
\[
\sup_{x}|\rho_{y}(x)|\leq Cy^{-\underline{\beta}},\qquad\int\rho_{y}(x)f_{X}^{0}(x)\,dx=0
\]
uniformly for sufficiently large $y$. Thus, 
\begin{equation}
\chi^{2}\!\left(f_{X\mid Y=y}\|f_{X}^{0}\right)\leq Cy^{-2\underline{\beta}}\label{eq:uniform-threshold-mark-chi2}
\end{equation}
uniformly for all sufficiently large $y$.

We first give a bound that holds for all $w\in\mathcal{W}_{n}$. Write
$I_{n}(w):=\{i:Y_{i}>w\}$ and $N_{n}(w):=|I_{n}(w)|=\sum_{i=1}^{n}\ind\{Y_{i}>w\}$.
Conditional on $(Y_{1},\ldots,Y_{n})$, the selected covariates are
independent. Their actual and independent-reference joint densities
are 
\[
p_{n,w}(x_{i}:i\in I_{n}(w)):=\prod_{i\in I_{n}(w)}f_{X\mid Y=Y_{i}}(x_{i}),\quad p_{n,w}^{0}(x_{i}:i\in I_{n}(w)):=\prod_{i\in I_{n}(w)}f_{X}^{0}(x_{i}).
\]
The reference density $p_{n,w}^{0}$ keeps the observed $Y$-values
and replaces the selected covariates by independent draws from $f_{X}^{0}$.
By multiplying these independent density ratios and \eqref{eq:uniform-threshold-mark-chi2},
for every $w\in\mathcal{W}_{n}$, 
\[
1+\chi^{2}(p_{n,w}\|p_{n,w}^{0})=\prod_{i\in I_{n}(w)}\{1+\chi^{2}(f_{X\mid Y=Y_{i}}\|f_{X}^{0})\}\leq\exp\{CN_{n}(\underline{w}_{n})\underline{w}_{n}^{-2\underline{\beta}}\}.
\]
Here each selected $Y_{i}$ exceeds $w\geq\underline{w}_{n}$, and
$N_{n}(w)\leq N_{n}(\underline{w}_{n})$. The bound does not depend
on $w$, so, conditional on $(Y_{1},\ldots,Y_{n})$, it holds simultaneously:
\[
\sup_{w\in\mathcal{W}_{n}}\chi^{2}(p_{n,w}\|p_{n,w}^{0})\leq\exp\left\{ CN_{n}(\underline{w}_{n})\underline{w}_{n}^{-2\underline{\beta}}\right\} -1.
\]

Because $n\Pp(Y>\overline{w}_{n})\to\infty$ and $\Pp(Y>\underline{w}_{n})\geq\Pp(Y>\overline{w}_{n})$,
\[
\frac{N_{n}(\underline{w}_{n})}{n\Pp(Y>\underline{w}_{n})}\xrightarrow{p}1.
\]
Condition \eqref{eq:uniform-threshold-rate} therefore implies 
\[
N_{n}(\underline{w}_{n})\underline{w}_{n}^{-2\underline{\beta}}\xrightarrow{p}0.
\]
The chi-square bound on total variation now gives 
\[
\sup_{w\in\mathcal{W}_{n}}\|p_{n,w}-p_{n,w}^{0}\|_{\mathrm{TV}}\leq\left[\frac{\exp\{CN_{n}(\underline{w}_{n})\underline{w}_{n}^{-2\underline{\beta}}\}-1}{2}\right]^{1/2}\xrightarrow{p}0.
\]
In particular, if $\widetilde{w}_{n}$ is measurable with respect
to $(Y_{1},\ldots,Y_{n})$ and satisfies \eqref{eq:ymeasurable-threshold},
then on $\{\widetilde{w}_{n}\in\mathcal{W}_{n}\}$ we may substitute
$w=\widetilde{w}_{n}$ in these uniform bounds.

Under the reference experiment, conditional on the observed $Y$-values
and on the threshold rule, the selected covariate marks are iid from
$f_{X}^{0}$ and are independent of the selected exceedance severities.
Thus their labels are exchangeable relative to the log exceedances,
so ordinary permutation calibration is conditionally exact up to permutation
discreteness. Also, on $\{\widetilde{w}_{n}\leq\overline{w}_{n}\}$
the selected sample size is at least $N_{n}(\overline{w}_{n})$, and
$n\Pp(Y>\overline{w}_{n})\to\infty$ implies $N_{n}(\overline{w}_{n})\to_{p}\infty$.
The discreteness error therefore vanishes uniformly over the range.

Total variation bounds the difference in rejection probabilities between
the actual and reference experiments. This proves the claim for any
$Y$-measurable threshold satisfying \eqref{eq:ymeasurable-threshold}.
The same uniform bound, applied to every deterministic $w\in\mathcal{W}_{n}$,
yields \eqref{eq:uniform-permutation-size}.

For the order-statistic assertion, the Hall expansion under $H_{0}$
implies $\bar{F}(y)\sim\mu_{\kappa}y^{-\alpha_{0}}$. Thus $u_{n}\to\infty$
and $\bar{F}(cu_{n})/\bar{F}(u_{n})\to c^{-\alpha_{0}}$ for every
fixed $c>0$. Choose the deterministic range $\mathcal{W}_{n}^{*}:=[u_{n}/2,2u_{n}]$.
It satisfies 
\[
n\bar{F}(2u_{n})\sim2^{-\alpha_{0}}K_{n}\to\infty,\qquad\sqrt{n\bar{F}(u_{n}/2)}\,(u_{n}/2)^{-\underline{\beta}}\sim2^{\alpha_{0}/2+\underline{\beta}}\sqrt{K_{n}}\,u_{n}^{-\underline{\beta}}\to0.
\]
Binomial concentration at the two endpoints gives 
\[
\frac{N_{n}(u_{n}/2)}{K_{n}}\to_{p}2^{\alpha_{0}}>1,\qquad\frac{N_{n}(2u_{n})}{K_{n}}\to_{p}2^{-\alpha_{0}}<1.
\]
Thus, $\Pp\{\widehat{w}_{n}\in\mathcal{W}_{n}^{*}\}\to1$. The threshold
$\widehat{w}_{n}$ is measurable with respect to the $Y$-values and
selects exactly $K_{n}$ observations under upper-tail continuity.
The outcome-measurable threshold result proved above therefore establishes
its asymptotic level.

Under the exact-Pareto benchmark $L(y;x)\equiv1$ for $y\geq y_{0}$,
\eqref{eq:uniform-threshold-mark-rn} has $B(y,x)=\alpha_{0}$ and
$f_{X\mid Y=y}=f_{X}^{0}$ for $y\geq y_{0}$. On $\{\widehat{w}_{n}\geq y_{0}\}$,
conditional on the observed $Y$-values, the selected covariates are
iid from $f_{X}^{0}$. The permutation test is then conditionally
finite-sample valid, up to permutation discreteness. \end{proof}

\noindent\textbf{Proof of Proposition~\ref{prop:local-window-permutation}.}

\begin{proof}Let $S_{n,j}$ be the selected indices in group $j$,
$j=1,2$. Condition on the covariates, the local thresholds, and the
boundary and selected identities. Denote this information by $\mathcal{G}_{n}$.
Given $\mathcal{G}_{n}$, the selected outcomes are independent: fixing
the boundaries and identities imposes only separate above- or below-threshold
restrictions. For $i\in S_{n,j}$, the conditional distribution of
$\log(Y_{i}/\widehat{w}_{n,j})$ is $F_{X_{i},\widehat{w}_{n,j}}$,
where $F_{x,w}$ is the distribution of $\log(Y/w)$ given $X=x,Y>w$,
with density $f_{x,w}$.

Let $F_{0}:=\mathrm{Exp}(\alpha_{0})$ have density $f_{0}(z)=\alpha_{0}e^{-\alpha_{0}z}$.
The uniform Hall expansion \eqref{eq:density-expansion} gives $\sup_{x}\chi^{2}(F_{x,w}\|F_{0})\leq Cw^{-2\underline{\beta}}$
for all sufficiently large $w$. Conditional on $\mathcal{G}_{n}$,
set 
\[
p_{n}(z_{i}:i\in\cup_{j=1}^{2}S_{n,j}):=\prod_{j=1}^{2}\prod_{i\in S_{n,j}}f_{X_{i},\widehat{w}_{n,j}}(z_{i}),\qquad p_{n}^{0}(z_{i}:i\in\cup_{j=1}^{2}S_{n,j}):=\prod_{j=1}^{2}\prod_{i\in S_{n,j}}f_{0}(z_{i}).
\]
On the high-probability event in \eqref{eq:local-window-diverging-thresholds},
the product formula from the proof of Proposition~\ref{prop:permutation}
gives 
\[
1+\chi^{2}(p_{n}\|p_{n}^{0})=\prod_{j=1}^{2}\prod_{i\in S_{n,j}}\{1+\chi^{2}(f_{X_{i},\widehat{w}_{n,j}}\|f_{0})\}\leq\exp\!\left\{ C\sum_{j=1}^{2}K_{n,j}\widehat{w}_{n,j}^{-2\underline{\beta}}\right\} \longrightarrow_{p}1.
\]
The chi-square-to-total-variation bound from that proof gives $\|p_{n}-p_{n}^{0}\|_{\mathrm{TV}}\to_{p}0$.
Under $p_{n}^{0}$, the pooled exceedances are iid $F_{0}$ and exchangeable
across the two groups. As $K_{n}\to\infty$, permutation discreteness
vanishes. Total variation then proves the level claim. The sufficient
rate follows from $\sum_{j=1}^{2}K_{n,j}(\widehat{w}_{n,j})^{-2\underline{\beta}}\leq K_{n}(\min_{j=1,2}\widehat{w}_{n,j})^{-2\underline{\beta}}$.
Under exact Pareto tails, $F_{x,w}=F_{0}$ for $w\geq y_{0}$, giving
conditional finite-sample validity. \end{proof}

\subsection{Tail-design concentration}

\label{app:tail-design}

\begin{lemma}[Selected covariates in a fixed group]\label{lem:group-tail-design}
Under the fixed-group conditions of Theorem~\ref{thm:degeneracy},
for every bounded continuous $v$ and each group $j=1,2$, 
\[
\frac{1}{K_{n,j}}\sum_{r\text{ retained in group }j}v(X_{r,n}^{\mathrm{LB}})\longrightarrow_{p}\Ee[v(X_{j}^{\infty})].
\]
For some $0<\delta'<\delta$, the selected covariates also satisfy
\[
\frac{1}{K_{n,j}}\sum_{r\text{ retained in group }j}\|X_{r,n}^{\mathrm{LB}}\|^{1+\delta'}=O_{p}(1).
\]
\end{lemma}

\noindent\textbf{Proof of Lemma~\ref{lem:group-tail-design}.}

\begin{proof}Let $Q_{j,>y}$ and $Q_{j,y}$ be the distributions
of $X$ in group $j$ given $Y>y$ and $Y=y$. Differentiating \eqref{eq:tail-model}
gives 
\[
\frac{dQ_{j,y}}{dQ_{j,>y}}(x)=\frac{h_{y}(x)}{\Ee_{Q_{j,>y}}h_{y}(X)},\qquad h_{y}(x):=\alpha(x)-y\partial_{y}\log L(y;x)=\alpha(x)+O(y^{-\underline{\beta}})
\]
uniformly in $x$. Tail concentration within the group gives $\alpha(X)\to\alpha_{j}^{*}$
under $Q_{j,>y}$. Thus $\|Q_{j,y}-Q_{j,>y}\|_{\mathrm{TV}}\to0$
and $Q_{j,y}\Rightarrow X_{j}^{\infty}$. By continuity of $\alpha$,
\[
\Pr(X_{j}^{\infty}\in\mathcal{M}_{j})=1.
\]
The density ratio is bounded, so $Q_{j,y}$ also has a uniform $(1+\delta)$
moment.

Given all outcomes and fixed group memberships, order observations
in group $j$ from largest to smallest outcome. Their covariates are
independent with distributions $Q_{j,Y_{(r)}}$. Let $\mathcal{H}_{n,j}$
collect these outcomes and group memberships. Choose $0<\delta'<\delta$
and $1<p\leq2$ such that $p(1+\delta')\leq1+\delta$. For either
$W_{r}=v(X_{(r)})$ or $W_{r}=\|X_{(r)}\|^{1+\delta'}$, the conditional
$p$th moments are uniformly bounded once the retained outcomes exceed
a sufficiently large fixed $u$. Since $\widehat{w}_{n,j}\to_{p}\infty$,
this event has probability tending to one. Writing $\xi_{r}:=W_{r}-\Ee(W_{r}\mid\mathcal{H}_{n,j})$,
conditional independence and the von Bahr--Esseen inequality give
\[
\Ee\!\left[\left|\sum_{r=1}^{K_{n,j}}\xi_{r}\right|^{p}\middle|\mathcal{H}_{n,j}\right]\leq CK_{n,j}.
\]
Hence, by Markov's inequality, 
\[
\Pr\!\left(\left|K_{n,j}^{-1}\sum_{r=1}^{K_{n,j}}\xi_{r}\right|>\epsilon\,\middle|\,\mathcal{H}_{n,j}\right)\leq C_{\epsilon}K_{n,j}^{1-p}\longrightarrow0.
\]
The conditional means of bounded $v$ converge to $\Ee[v(X_{j}^{\infty})]$,
while those of $\|X\|^{1+\delta'}$ remain bounded. Both claims follow.
\end{proof}

\noindent\textbf{Proof of Theorem~\ref{thm:degeneracy}.}

\begin{proof} \textbf{Step 1 (concentration and subsequential limits
of the selected design).} Under \eqref{eq:tail-model}, 
\[
f_{X,n}(x)=\frac{\kappa(x)w_{n}^{-\alpha(x)}L(w_{n};x)f_{X}(x)}{\int\kappa(v)w_{n}^{-\alpha(v)}L(w_{n};v)f_{X}(v)\,dv}.
\]
By Assumption~\ref{ass:hall}, $L(w_{n};x)\to1$ uniformly in $x$.
Fix $\epsilon>0$ and let $\mathcal{M}^{\epsilon}:=\{x:\alpha(x)<\alpha_{\min}+\epsilon\}$,
which has positive $P_{X}$-measure. Since $\kappa>0$ and $0<\Ee[\kappa(X)]<\infty$,
\[
\begin{split}P_{X,n}(\mathcal{X}\setminus\mathcal{M}^{\epsilon}) & \leq\frac{\Ee[\kappa(X)]w_{n}^{-(\alpha_{\min}+\epsilon)}\{1+o(1)\}}{\Ee[\kappa(X)\ind\{X\in\mathcal{M}^{\epsilon/2}\}]w_{n}^{-(\alpha_{\min}+\epsilon/2)}\{1+o(1)\}}\\
 & =O(w_{n}^{-\epsilon/2})\longrightarrow0.
\end{split}
\]
Thus $\alpha(X)\to\alpha_{\min}$ in probability under $P_{X,n}$.
Assumption~\ref{ass:dcov-moments} makes $\{P_{X,n}\}$ tight. Any
subsequence therefore has a further subsequence on which the selected
covariate converges in distribution to a random variable $X_{\infty}^{G}$.
Continuity of $\alpha$ on the support implies $\Pr(X_{\infty}^{G}\in\mathcal{M})=1$.
The distribution of $X_{\infty}^{G}$ may differ across subsequences.
For related concentration results, see \citet[Thm.~2.1 and Cor.~2.4]{Yang2026}.

\textbf{Step 2 (conditional distribution of the mark).} Without imposing
$H_{0}$, Assumption~\ref{ass:hall} gives, uniformly in $x$ and
$z\geq0$, 
\[
\Pr(Z^{*}>z\mid X=x,Y>w_{n})=e^{-\alpha(x)z}\frac{L(w_{n}e^{z};x)}{L(w_{n};x)}=e^{-\alpha(x)z}\{1+O(w_{n}^{-\underline{\beta}})\}.
\]
Because $\alpha$ is bounded away from zero, this also provides a
uniform exponential envelope for the marks. Integrating the survival
functions shows, for each fixed $t$, that their conditional characteristic
functions equal $\alpha(x)/(\alpha(x)-it)+o(1)$ uniformly in $x$.
Step 1 and boundedness of $\alpha$ then imply 
\[
\phi_{n}(s,t)-\phi_{X,n}(s)\frac{\alpha_{\min}}{\alpha_{\min}-it}\longrightarrow0\qquad\text{for every }(s,t).
\]
Along the subsequence of Step~1, $(X,Z^{*})$ conditional on $Y>w_{n_{\ell}}$
converges in distribution to $(X_{\infty}^{G},E_{\infty})$, where
$E_{\infty}$ is independent of $X_{\infty}^{G}$ and has exponential
rate $\alpha_{\min}$.

\textbf{Step 3 (continuity of the dCov functional).} Suppose $P_{n}\Rightarrow P_{0}$,
$\sup_{n}\Ee_{P_{n}}\|X\|^{1+\delta}<\infty$, and the nonnegative
mark $Z$ has a uniform exponential envelope. Hölder's inequality
with conjugate exponents $1+\delta$ and $(1+\delta)/\delta$ gives
\[
\sup_{n}\Ee_{P_{n}}[\|X\|Z]\leq\sup_{n}\left(\Ee_{P_{n}}\|X\|^{1+\delta}\right)^{1/(1+\delta)}\left(\Ee_{P_{n}}Z^{(1+\delta)/\delta}\right)^{\delta/(1+\delta)}<\infty.
\]
Write $a_{s}(X):=e^{is'X}-1$ and $b_{t}(Z):=e^{itZ}-1$. Since $\mathscr{C}_{P_{n}}(s,t)=\Ee_{P_{n}}[a_{s}b_{t}]-\Ee_{P_{n}}[a_{s}]\Ee_{P_{n}}[b_{t}]$,
\[
|\mathscr{C}_{P_{n}}(s,t)|\leq\|s\||t|\left\{ \Ee_{P_{n}}[\|X\|Z]+\Ee_{P_{n}}\|X\|\,\Ee_{P_{n}}Z\right\} .
\]
Using $|e^{iu}-1|\leq\min\{2,|u|\}$ also gives the separate bounds
$4\|s\|\Ee_{P_{n}}\|X\|$, $4|t|\Ee_{P_{n}}Z$, and $2$. Combining
the four bounds, for a constant $C$ independent of $n$, 
\[
|\mathscr{C}_{P_{n}}(s,t)|\leq C\min\{1,\|s\|\}\min\{1,|t|\}.
\]
The squared integrand in \eqref{eq:population-dcov} is dominated
by 
\[
C^{2}\frac{\min\{1,\|s\|\}^{2}\min\{1,|t|\}^{2}}{\|s\|^{d+1}|t|^{2}},
\]
which is integrable. In polar coordinates the radial $s$-integral
is 
\[
\int_{0}^{\infty}\min\{1,r\}^{2}r^{-2}\,dr=2<\infty,
\]
and the $t$-integral is finite by the same calculation. Weak convergence
gives pointwise convergence of the characteristic functions. Dominated
convergence therefore shows that the weighted integral of $|\mathscr{C}_{P_{n}}-\mathscr{C}_{P_{0}}|^{2}$
in \eqref{eq:population-dcov} tends to zero.

Apply this result along each convergent subsequence identified in
Steps 1--2. Its limiting coordinates are independent, so the corresponding
subsequence of $\Delta_{n}^{A}$ tends to zero. If \eqref{eq:global-signal-vanishes}
failed, there would be a subsequence bounded away from zero, and tightness
would supply a further subsequence with a product limit, a contradiction.
Thus \eqref{eq:global-signal-vanishes} holds along the full sequence,
without uniqueness of the selected covariate limit.

\textbf{Step 4 (fixed groups and local selection).} Lemma~\ref{lem:group-tail-design}
and Step~2 show that within group $j$ the conditional selected-pair
distribution approaches $X_{j}^{\infty}\otimes\Exp(\alpha_{j}^{*})$.
Since $K_{n,j}/K_{n}\to\rho_{j}$, the pooled distribution approaches
the mixture in the theorem. The lemma's moment bound and Step~3 give
\eqref{eq:local-nondegenerate-limit}.

\textbf{Step 5 (positive local signal).} By the fixed-group concentration
argument above, $\Pr(X_{j}^{\infty}\in\mathcal{M}_{j})=1$ for $j=1,2$.
If $\alpha_{1}^{*}\neq\alpha_{2}^{*}$, then $\mathcal{M}_{1}\cap\mathcal{M}_{2}=\varnothing$
by definition, so the group label is determined by $X_{\infty}$ almost
surely. Since $Z_{\infty}$ has exponential rate $\alpha_{j}^{*}$
conditional on group $j$, the unequal rates imply that $(X_{\infty},Z_{\infty})$
are dependent. Their first moments are finite, so the characterization
of \citet{SzekelyRizzoBakirov2007} yields $\Delta_{\infty}^{\mathrm{LB}}>0$.

\end{proof}

\subsection{Global signal decay and consistency}

\label{app:degeneracy-rate}

Theorem~\ref{thm:degeneracy} establishes that the global population
signal tends to zero under fixed alternatives. Consistency depends
on its rate relative to the number of tail sample observations. The
following regular benchmarks illustrate logarithmic signal decay in
$w_{n}$, with an exponent determined by the local geometry of $\alpha(\cdot)$
at its minimizers. There is no universal logarithmic rate. We begin
with the canonical scalar design of \citet{Yang2026}. There the tail
index itself, $\alpha(X)=\exp(X)$, is taken to be uniform on $[u,\bar{u}]$,
so that $X$ has support $[\log u,\log\bar{u}]$ with density bounded
away from zero and infinity, and the conditional density calculation
underlying Corollary~2.2 in \citet{Yang2026} implies 
\[
(\alpha(X)-u)\log w_{n}\Rightarrow\operatorname{Exp}(1).
\]
Hence $X-\log u=O_{p}((\log w_{n})^{-1})$, consistent with $\operatorname{Var}(X\mid Y>w_{n})\asymp(\log w_{n})^{-2}$.
Let $h_{n}:=(\log w_{n})^{-1}$ and $\widetilde{X}_{n}:=(X-\log u)/h_{n}$.
Squared distance covariance is homogeneous of degree one in its first
argument, so 
\[
\Delta_{n}^{A}=h_{n}\,\operatorname{dCov}^{2}(\widetilde{X}_{n},Z^{*}).
\]
A first-order expansion of the conditional $\operatorname{Exp}(\alpha(X))$
distribution around $u$ makes the dependence between $\widetilde{X}_{n}$
and $Z^{*}$ of order $h_{n}$. If the corresponding first-order dCov
coefficient is nonzero, then 
\[
\Delta_{n}^{A}\asymp h_{n}^{3}=(\log w_{n})^{-3}.
\]
Thus, under the regularity conditions of Theorem~\ref{thm:global-power},
global consistency in this benchmark holds if $K_{n}/(\log w_{n})^{3}\to\infty$.

The exponent depends on the local geometry, and the boundary case
above is the less favorable of the two leading configurations. If
instead $\alpha(\cdot)$ has a smooth interior minimum at $x^{*}$
with $\alpha''(x^{*})>0$, the tilt $e^{-\alpha(x)\log w_{n}}$ is
locally Gaussian, so the concentration scale is $h_{n}:=(\log w_{n})^{-1/2}$
rather than $(\log w_{n})^{-1}$. The induced dependence is then of
order $\alpha(X)-\alpha_{\min}\asymp h_{n}^{2}$, giving $\dCov^{2}(\widetilde{X}_{n},Z^{*})\asymp h_{n}^{4}$
and 
\[
\Delta_{n}^{A}\asymp h_{n}^{5}=(\log w_{n})^{-5/2}.
\]
More generally, logarithmic signal decay of the form $\Delta_{n}^{A}\asymp(\log w_{n})^{-q}$
with finite $q>0$ is accommodated whenever $K_{n}/(\log w_{n})^{q}\to\infty$.
For example, if $K_{n}\asymp n^{b}$ for some $0<b<1$ and $\log w_{n}=O(\log n)$,
this signal condition is satisfied. Thus, in these logarithmic-rate
benchmarks, a shrinking global population signal is compatible with
consistency under the stated regularity conditions. The simulation
comparison in Table~\ref{tab:global-local} shows that local balancing
can nevertheless provide substantial finite-sample gains when fixed
heterogeneity produces strong concentration. Such gains concern relative
power at the reported sample sizes and do not establish inconsistency
of the global test.

\subsection{Sample calculation and consistency proofs}

\label{app:power-analysis}

\subsubsection{dCov sample tightness}

For $r$ draws from $P$, write $\widehat{\mathscr{C}}_{P}$ for the
empirical version of $\mathscr{C}_{P}$. The sample size is understood
from context. For the selected global and locally balanced samples,
denote these empirical differences by $\widehat{\mathscr{C}}_{n}$
and $\widehat{\mathscr{C}}_{n}^{\mathrm{LB}}$. For a function $v(s,t)$,
define the squared dCov norm 
\[
\|v\|_{\mathrm{dCov}}^{2}:=\frac{1}{c_{d}c_{1}}\int_{\R^{d}}\int_{\R}\frac{|v(s,t)|^{2}}{\|s\|^{d+1}|t|^{2}}\,dt\,ds.
\]

For the locally balanced sample, define 
\begin{align}
\bar{\phi}_{X,n}^{\mathrm{LB}}(s) & :=K_{n}^{-1}\sum_{r=1}^{K_{n}}e^{is'X_{r,n}^{\mathrm{LB}}},\nonumber \\
\mathscr{C}_{n}^{\mathrm{LB},A}(s,t) & :=K_{n}^{-1}\sum_{r=1}^{K_{n}}\{e^{is'X_{r,n}^{\mathrm{LB}}}-\bar{\phi}_{X,n}^{\mathrm{LB}}(s)\}\varphi_{r,n}^{A}(t).\label{eq:local-balanced-signal}
\end{align}
Here $\varphi_{r,n}^{A}(t):=\Ee_{A}[e^{itZ_{r,n}^{\mathrm{LB}}}\mid\mathcal{G}_{n}]$.
The conditional joint and marginal characteristic functions of the
uniformly chosen pair in \eqref{eq:local-balanced-dcov-signal} give
\[
\Delta_{n}^{\mathrm{LB},A}=\|\mathscr{C}_{n}^{\mathrm{LB},A}\|_{\mathrm{dCov}}^{2}.
\]

\begin{lemma}[dCov sample tightness]\label{lem:dcov-triangular}
Let $P_{n}$ satisfy the moment bound \eqref{eq:process-moment},
and let $r_{n}\to\infty$. For $r_{n}$ independent draws from $P_{n}$,
\begin{equation}
r_{n}\|\widehat{\mathscr{C}}_{P_{n}}-\mathscr{C}_{P_{n}}\|_{\mathrm{dCov}}^{2}=O_{p}(1).\label{eq:dcov-process-tight}
\end{equation}
\end{lemma}

\begin{proof} Write $r:=r_{n}$, and let $\widehat{\phi}_{X,P_{n}}$
and $\widehat{\phi}_{Z,P_{n}}$ be the empirical marginal characteristic
functions from these $r$ draws. Put $a_{i}(s):=e^{is'X_{i}}-\phi_{X,P_{n}}(s)$
and $b_{i}(t):=e^{itZ_{i}}-\phi_{Z,P_{n}}(t)$. Expanding empirical
characteristic functions gives the exact identity 
\[
\widehat{\mathscr{C}}_{P_{n}}-\mathscr{C}_{P_{n}}=\frac{1}{r}\sum_{i=1}^{r}\{a_{i}(s)b_{i}(t)-\Ee[a_{i}(s)b_{i}(t)]\}-\{\widehat{\phi}_{X,P_{n}}(s)-\phi_{X,P_{n}}(s)\}\{\widehat{\phi}_{Z,P_{n}}(t)-\phi_{Z,P_{n}}(t)\}.
\]
The bound $|e^{iu}-1|\leq\min\{2,|u|\}$ and the moment assumption
give a uniform bound on the expected squared norm of the first term
after multiplication by $r$. For the second term, the squared norm
factorizes into two ordinary integrals. The distance identity of \citet{SzekelyRizzoBakirov2007}
shows its expected contribution after multiplication by $r$ tends
to zero. This proves \eqref{eq:dcov-process-tight}. \end{proof}

\subsubsection{Consistency proof}

\noindent\textbf{Proof of Theorem~\ref{thm:global-power}.}

\begin{proof}For the global test, the selected pairs are iid conditional
on $N_{n}$. Lemma~\ref{lem:dcov-triangular} and $N_{n}/K_{n}\to_{p}1$
give 
\[
N_{n}\|\widehat{\mathscr{C}}_{n}-\mathscr{C}_{n}^{A}\|_{\mathrm{dCov}}^{2}=O_{p}(1),
\]
where $\mathscr{C}_{n}^{A}(s,t):=\phi_{n}^{A}(s,t)-\phi_{X,n}^{A}(s)\phi_{Z,n}^{A}(t)$.
the $\phi^{A}$ symbols are the joint and marginal characteristic
functions of the selected pair under the alternative. Applying the
usual triangle inequality to the square roots of the squared norms
shows 
\[
\sqrt{N_{n}\widehat{\Delta}_{n}}\geq\sqrt{N_{n}\Delta_{n}^{A}}-O_{p}(1)\to_{p}\infty.
\]

The permutation critical value remains bounded on this scale. If $\bar{d}_{X,n}$
and $\bar{d}_{Z,n}$ are the average pairwise distances within the
selected covariates and log exceedances, a sampling-without-replacement
calculation gives 
\begin{equation}
\Ee_{\pi}[N_{n}\widehat{\Delta}_{n}^{\pi}\mid\text{selected sample}]=\frac{N_{n}}{N_{n}-1}\bar{d}_{X,n}\bar{d}_{Z,n}=O_{p}(1).\label{eq:perm-mean-dcov}
\end{equation}
Markov's inequality gives $N_{n}c_{n,1-\eta}^{\pi}=O_{p}(1)$, proving
\eqref{eq:fixed-consistency}.

For the locally balanced test, condition on $\mathcal{G}_{n}$ and
let $a_{r,n}(s):=e^{is'X_{r,n}^{\mathrm{LB}}}-\bar{\phi}_{X,n}^{\mathrm{LB}}(s)$
and $\varepsilon_{r,n}(t):=e^{itZ_{r,n}^{\mathrm{LB}}}-\varphi_{r,n}^{A}(t)$.
The marks $\varepsilon_{r,n}$ are independent and centered given
$\mathcal{G}_{n}$, and 
\begin{equation}
\widehat{\mathscr{C}}_{n}^{\mathrm{LB}}-\mathscr{C}_{n}^{\mathrm{LB},A}=K_{n}^{-1}\sum_{r=1}^{K_{n}}a_{r,n}\varepsilon_{r,n}.\label{eq:local-balanced-process-decomp}
\end{equation}
The groupwise moment condition \eqref{eq:local-balanced-process-moment}
and $|e^{iu}-1|\leq\min\{2,|u|\}$ bound the conditional expectation
of the integrated squared sum. Hence 
\[
K_{n}\|\widehat{\mathscr{C}}_{n}^{\mathrm{LB}}-\mathscr{C}_{n}^{\mathrm{LB},A}\|_{\mathrm{dCov}}^{2}=O_{p}(1).
\]
The triangle inequality now gives $\sqrt{K_{n}\widehat{\Delta}_{n}^{\mathrm{LB}}}\geq\sqrt{K_{n}\Delta_{n}^{\mathrm{LB},A}}-O_{p}(1)\to_{p}\infty$.
Condition \eqref{eq:local-balanced-process-moment} also bounds the
average pairwise distances in probability. The calculation in \eqref{eq:perm-mean-dcov},
with $K_{n}$ replacing $N_{n}$, gives $\Ee_{\pi}[K_{n}\widehat{\Delta}_{n}^{\mathrm{LB},\pi}\mid\text{pooled sample}]=O_{p}(1)$.
Markov's inequality yields $K_{n}c_{n,1-\eta}^{\mathrm{LB},\pi}=O_{p}(1)$
and proves \eqref{eq:local-balanced-consistency}. \end{proof}

\subsubsection{Global power with an order-statistic threshold}

\label{app:global-orderstat-power}

\begin{corollary}[Global power with an order-statistic threshold]\label{cor:global-orderstat-power}
Suppose \eqref{eq:tail-model} and Assumption~\ref{ass:hall} hold
under a fixed alternative or the local sequence \eqref{eq:local-alt},
$0<\Ee[\kappa(X)]<\infty$, and $Y$ has a continuous upper tail.
Let $K_{n}\to\infty$, $K_{n}/n\to0$, $\widehat{w}_{n}:=Y_{n-K_{n}:n}$,
and choose $u_{n}$ so that $\Pr_{A}(Y>u_{n})=K_{n}/n$. Put $\mathcal{W}_{n}=[u_{n}/2,2u_{n}]$
and assume, for some $\delta>0$, 
\begin{equation}
\sup_{n}\sup_{w\in\mathcal{W}_{n}}\Ee_{A}\!\left[(\|X\|+\log(Y/w))^{2+\delta}\mid Y>w\right]<\infty.\label{eq:orderstat-power-moments}
\end{equation}
Let $\Delta_{n}^{A}(w):=\dCov^{2}(X,\log(Y/w)\mid Y>w)$ and compute
the test from the $K_{n}$ observations above $\widehat{w}_{n}$. 
\begin{enumerate}
\item[(i)] Under a fixed alternative, 
\begin{equation}
K_{n}\Delta_{n}^{A}(\widehat{w}_{n})\to_{p}\infty\quad\Longrightarrow\quad\Pr_{A}\{\widehat{\Delta}_{n}>c_{n,1-\eta}^{\pi}\}\to1.\label{eq:orderstat-fixed-power}
\end{equation}
It is sufficient that $K_{n}\inf_{w\in\mathcal{W}_{n}}\Delta_{n}^{A}(w)\to\infty$.
Alternatively, $K_{n}\Delta_{n}^{A}(u_{n})\to\infty$ suffices if
$\Delta_{n}^{A}(\widehat{w}_{n})/\Delta_{n}^{A}(u_{n})\to_{p}1$. 
\item[(ii)] Under \eqref{eq:local-alt}, if 
\begin{equation}
\sqrt{K_{n}}\,u_{n}^{-\underline{\beta}}\to0,\qquad\log u_{n}/\sqrt{K_{n}}\to0,\label{eq:orderstat-local-power-rates}
\end{equation}
then 
\begin{equation}
K_{n}\widehat{\Delta}_{n}\Longrightarrow T_{G}(c).\label{eq:orderstat-local-power}
\end{equation}
If $q_{G,1-\eta}$ is unique and is a continuity point of the distributions
of both $T_{G}(0)$ and $T_{G}(c)$, then $K_{n}c_{n,1-\eta}^{\pi}\to_{p}q_{G,1-\eta}$
and the rejection probability converges to $\Pr\{T_{G}(c)>q_{G,1-\eta}\}$. 
\end{enumerate}
\end{corollary}

\noindent\textbf{Proof of Corollary~\ref{cor:global-orderstat-power}.}

\begin{proof} Under a fixed alternative, if $a:=\inf\alpha(X)$,
the Hall expansion gives $\Pr_{A}(Y>tu_{n})/\Pr_{A}(Y>u_{n})\to t^{-a}$.
Under \eqref{eq:local-alt}, the same ratio tends to $t^{-\alpha_{0}}$.
Hence binomial concentration implies $\widehat{w}_{n}/u_{n}\to_{p}1$
and $\Pr_{A}(\widehat{w}_{n}\in\mathcal{W}_{n})\to1$.

Conditional on $\widehat{w}_{n}=w$ and on the selected identities,
the retained pairs are iid from 
\[
P_{n}(w):=\mathcal{L}\{(X,\log(Y/w))\mid Y>w\}.
\]
The moment bound makes Lemma~\ref{lem:dcov-triangular} uniform over
$\mathcal{W}_{n}$. This gives 
\[
K_{n}\|\widehat{\mathscr{C}}_{n}-\mathscr{C}_{n,\widehat{w}_{n}}^{A}\|_{\mathrm{dCov}}^{2}=O_{p}(1),\qquad K_{n}c_{n,1-\eta}^{\pi}=O_{p}(1),
\]
so the triangle inequality proves part (i).

For part (ii), \eqref{eq:orderstat-local-power-rates} and the Hall
expansion hold uniformly over $w\in\mathcal{W}_{n}$. The characteristic-function
calculation in the proof of Theorem~\ref{thm:local-power} therefore
yields 
\[
\sup_{w\in\mathcal{W}_{n}}\|\sqrt{K_{n}}\mathscr{C}_{n,w}^{A}-cH_{g,f_{X}^{0}}\|_{\mathrm{dCov}}^{2}\to0.
\]
The same uniform moment bound gives the centered-process and permutation
limits of Lemma~\ref{lem:dcov-local-shift}. Conditioning on $\widehat{w}_{n}$
proves part (ii). \end{proof}

\section{Local-Power Details and Proof}

\label{app:local-power}

This appendix gives the additional rate and design conditions, process
definitions, and proof for Theorem~\ref{thm:local-power}. The argument
also covers $c=0$, which yields the explicit null limits stated in
the main text.

\subsection{Setup and limit}

\label{app:local-power-setup}

Assume the Hall and relevant moment conditions hold uniformly along
\eqref{eq:local-alt}. For the global test, require 
\[
\sqrt{K_{n}}w_{n}^{-\underline{\beta}}\to0,\qquad\log w_{n}/\sqrt{K_{n}}\to0.
\]
For the locally balanced test, impose \eqref{eq:local-balanced-process-moment}
uniformly, assume $g$ is continuous outside a set of zero limiting
probability, and require 
\begin{equation}
\sqrt{K_{n}}\max_{j=1,2}\widehat{w}_{n,j}^{-\underline{\beta}}\to_{p}0,\qquad\max_{j=1,2}\log\widehat{w}_{n,j}/\sqrt{K_{n}}\to_{p}0.\label{eq:local-balanced-power-rates}
\end{equation}
Suppose $K_{n,j}/K_{n}\to\rho_{j}>0$, $\rho_{1}+\rho_{2}=1$, and
\begin{equation}
K_{n}^{-1}\sum_{r=1}^{K_{n}}v(X_{r,n}^{\mathrm{LB}})\to_{p}\int v(x)f_{X}^{\mathrm{LB}}(x)\,dx\label{eq:local-balanced-design-limit}
\end{equation}
for every bounded continuous $v$, with uniform integrability of order
$2+\delta$. If the learned groups converge to fixed limits $\mathcal{C}_{1},\mathcal{C}_{2}$,
then 
\[
f_{X}^{\mathrm{LB}}(x)=\sum_{j=1}^{2}\rho_{j}\frac{\kappa(x)f_{X}(x)\ind\{x\in\mathcal{C}_{j}\}}{\int_{\mathcal{C}_{j}}\kappa(v)f_{X}(v)\,dv},
\]
equal allocation gives $\rho_{1}=\rho_{2}=1/2$.

For a covariate density $f$, let $\phi_{f}(s)=\int e^{is'x}f(x)\,dx$,
$\phi_{0}(t)=\alpha_{0}/(\alpha_{0}-it)$, and 
\[
H_{g,f}(s,t):=-\frac{it}{(\alpha_{0}-it)^{2}}\left\{ \int e^{is'x}g(x)f(x)\,dx-\phi_{f}(s)\int g(x)f(x)\,dx\right\} .
\]
Let $G_{f}$ denote the centered Gaussian limit of the empirical characteristic-function
process for independent $X\sim f$ and $Z\sim\Exp(\alpha_{0})$, as
in \citet[Theorem~5]{SzekelyRizzoBakirov2007}. With the dCov norm
introduced in Appendix~\ref{app:power-analysis}, these objects give
the main-text limits 
\[
T_{f}(c)=\|G_{f}+cH_{g,f}\|_{\mathrm{dCov}}^{2},\qquad T_{G}(c)=T_{f_{X}^{0}}(c),\qquad T_{\mathrm{LB}}(c)=T_{f_{X}^{\mathrm{LB}}}(c).
\]

\begin{lemma}[Local-shift and permutation limits]\label{lem:dcov-local-shift}
Suppose the conditions of Lemma~\ref{lem:dcov-triangular} hold,
$P_{n}$ converges to independent $X\sim f$ and $Z\sim\Exp(\alpha_{0})$,
and $\|\sqrt{r_{n}}\mathscr{C}_{P_{n}}-h\|_{\mathrm{dCov}}^{2}\to0$.
Then 
\begin{equation}
r_{n}\|\widehat{\mathscr{C}}_{P_{n}}\|_{\mathrm{dCov}}^{2}\Longrightarrow\|G_{f}+h\|_{\mathrm{dCov}}^{2}.\label{eq:dcov-shift-lemma}
\end{equation}
If the empirical marginals converge with the same moment bound, the
permutation distribution converges to that of $\|G_{f}\|_{\mathrm{dCov}}^{2}$.
\end{lemma}

\begin{proof} Apply the multivariate CLT to the centered empirical
characteristic-function process on bounded frequency sets. Lemma~\ref{lem:dcov-triangular}
and the same moment bounds control the neighborhoods of zero and the
frequency tails, while the deterministic shift gives \eqref{eq:dcov-shift-lemma}.
Conditional sampling-without-replacement gives the same Gaussian limit
for the permuted process. \end{proof}

\noindent\textbf{Proof of Theorem~\ref{thm:local-power}.}

\begin{proof} For the global test, let $f_{X,n}^{0}$ and $f_{X,n}^{c}$
be the densities of $X\mid Y>w_{n}$ under the null and under \eqref{eq:local-alt},
respectively. Bayes' formula gives 
\[
\frac{f_{X,n}^{c}(x)}{f_{X,n}^{0}(x)}=\frac{e^{-\delta_{n}g(x)\log w_{n}}}{\int e^{-\delta_{n}g(v)\log w_{n}}f_{X,n}^{0}(v)\,dv}.
\]
Because $g$ is bounded and $\delta_{n}\log w_{n}\to0$, the ratio
converges uniformly to one. The Hall expansion also gives $f_{X,n}^{0}\to f_{X}^{0}$
in $L^{1}$, and hence $f_{X,n}^{c}\to f_{X}^{0}$ in $L^{1}$.

Put $\phi_{\alpha}(t):=\alpha/(\alpha-it)$ and $\dot{\phi}_{0}(t):=\partial_{\alpha}\phi_{\alpha}(t)|_{\alpha=\alpha_{0}}=-it/(\alpha_{0}-it)^{2}$.
Let 
\[
\psi_{n}(x,t):=\Ee_{n,c}[e^{itZ^{*}}\mid X=x,Y>w_{n}].
\]
The Hall expansion, a first-order Taylor expansion of $\phi_{\alpha_{n}(x)}(t)$
around $\alpha_{0}$, and $\sqrt{K_{n}}w_{n}^{-\underline{\beta}}\to0$
give 
\[
\psi_{n}(x,t)=\phi_{0}(t)+\delta_{n}g(x)\dot{\phi}_{0}(t)+r_{n}(x,t),
\]
where the contribution of $r_{n}$ to the centered characteristic-function
process is $o(K_{n}^{-1/2})$ in the dCov norm. Since 
\[
\mathscr{C}_{n}^{A}(s,t)=\Ee_{n,c}\!\left[\{e^{is'X}-\phi_{X,n}^{c}(s)\}\psi_{n}(X,t)\,\middle|\,Y>w_{n}\right],
\]
the common term $\phi_{0}(t)$ cancels. Therefore 
\begin{align*}
\sqrt{K_{n}}\mathscr{C}_{n}^{A}(s,t)={} & c\,\dot{\phi}_{0}(t)\left\{ \int e^{is'x}g(x)f_{X,n}^{c}(x)\,dx-\phi_{X,n}^{c}(s)\int g(x)f_{X,n}^{c}(x)\,dx\right\} +R_{n}(s,t),
\end{align*}
where $\|R_{n}\|_{\mathrm{dCov}}\to0$. The $L^{1}$ convergence of
$f_{X,n}^{c}$, boundedness of $g$, and the moment condition then
imply 
\[
\|\sqrt{K_{n}}\mathscr{C}_{n}^{A}-cH_{g,f_{X}^{0}}\|_{\mathrm{dCov}}^{2}\to0.
\]
Also, the selected-pair distribution converges to the product of $X\sim f_{X}^{0}$
and $Z\sim\Exp(\alpha_{0})$. Conditional on $N_{n}$ the selected
pairs are iid, and $N_{n}/K_{n}\to_{p}1$. Lemma~\ref{lem:dcov-local-shift},
followed by Slutsky's theorem, therefore gives \eqref{eq:sample-local-power}.
Its permutation part gives convergence of the conditional permutation
distribution to that of $T_{G}(0)$. When $q_{G,1-\eta}$ is unique
and is a continuity point of that distribution, 
\[
K_{n}c_{n,1-\eta}^{\pi}\to_{p}q_{G,1-\eta}.
\]

For the locally balanced test, condition on $\mathcal{G}_{n}$. The
retained covariates are then fixed and the retained exceedance marks
are independent. The Hall expansion and \eqref{eq:local-balanced-power-rates}
give, uniformly in the weighted integrals defining the dCov norm,
\[
\varphi_{r,n}^{A}(t)=\phi_{0}(t)+\frac{c}{\sqrt{K_{n}}}g(X_{r,n}^{\mathrm{LB}})\dot{\phi}_{0}(t)+r_{r,n}(t),
\]
where the pooled contribution of the remainders is $o_{p}(K_{n}^{-1/2})$
in the dCov norm. Because $\sum_{r=1}^{K_{n}}\{e^{is'X_{r,n}^{\mathrm{LB}}}-\bar{\phi}_{X,n}^{\mathrm{LB}}(s)\}=0$,
the common $\phi_{0}$ term cancels exactly in \eqref{eq:local-balanced-signal}.
Hence 
\begin{align*}
\sqrt{K_{n}}\mathscr{C}_{n}^{\mathrm{LB},A}(s,t)={} & c\,\dot{\phi}_{0}(t)\bigg[K_{n}^{-1}\sum_{r=1}^{K_{n}}e^{is'X_{r,n}^{\mathrm{LB}}}g(X_{r,n}^{\mathrm{LB}})\\
 & \hspace{32mm}-\bar{\phi}_{X,n}^{\mathrm{LB}}(s)K_{n}^{-1}\sum_{r=1}^{K_{n}}g(X_{r,n}^{\mathrm{LB}})\bigg]+R_{n}^{\mathrm{LB}}(s,t),
\end{align*}
where $\|R_{n}^{\mathrm{LB}}\|_{\mathrm{dCov}}\to_{p}0$. By \eqref{eq:local-balanced-design-limit},
boundedness of $g$, and its assumed continuity outside a set of zero
limiting probability, 
\[
\|\sqrt{K_{n}}\mathscr{C}_{n}^{\mathrm{LB},A}-cH_{g,f_{X}^{\mathrm{LB}}}\|_{\mathrm{dCov}}^{2}\to_{p}0.
\]

Next use the decomposition in \eqref{eq:local-balanced-process-decomp}.
Conditional on $\mathcal{G}_{n}$, the Lindeberg--Feller CLT at each
finite set of frequencies, together with \eqref{eq:local-balanced-design-limit},
gives the centered Gaussian limit $G_{f_{X}^{\mathrm{LB}}}$. The
moment bound \eqref{eq:local-balanced-process-moment} controls the
neighborhoods of zero and the frequency tails exactly as in Lemma~\ref{lem:dcov-local-shift}.
The continuous mapping theorem therefore yields \eqref{eq:local-balanced-local-power-stat}.
The same sampling-without-replacement argument used for the permutation
part of Lemma~\ref{lem:dcov-local-shift} gives the conditional permutation
limit $T_{\mathrm{LB}}(0)$ and, under uniqueness and continuity,
\[
K_{n}c_{n,1-\eta}^{\mathrm{LB},\pi}\to_{p}q_{\mathrm{LB},1-\eta}.
\]

For either $a\in\{G,\mathrm{LB}\}$, combine the statistic limit with
the corresponding critical-value limit. If $T_{a}(c)$ is continuous
at $q_{a,1-\eta}$, Slutsky's theorem gives \eqref{eq:local-power-function}.
For $c=0$, this probability equals $\eta$. \end{proof} 
\end{document}